\documentclass[a4paper,11pt]{article}
\usepackage[utf8]{inputenc}
\usepackage{physics}
\usepackage[margin=2.54cm]{geometry}

\usepackage{tikz}
\usepackage[export]{adjustbox}
\usepackage{tabularx}
\usepackage{ccaption}
\usepackage{graphicx}
\usepackage{float}
\usepackage{cite}
\usepackage{subfig}
\usepackage{multirow}
\graphicspath{{Figures/}{../Figures/}}
\usepackage{subfiles}
\usepackage{wrapfig}
\usepackage{appendix}
\usepackage{mathtools}
\usepackage{amssymb}
\usepackage{amsmath}
\usepackage{xcolor}
\usepackage{qcircuit}
\usepackage{algorithm}
\usepackage{algpseudocode}
\usepackage{dsfont}
\usepackage{thm-restate}
\usepackage{enumitem} 
\usepackage{authblk}

\usepackage[pagebackref,unicode]{hyperref}
\usepackage[nameinlink,capitalize]{cleveref}

\hypersetup{
    colorlinks=true, 
    linkcolor=blue, 
    citecolor=blue, 
    urlcolor=blue 
}

\usepackage{thmtools,thm-restate}
\usepackage{amsthm}
\newtheorem{theorem}{Theorem}

\newtheorem{lemma}{Lemma}
\newtheorem{corollary}{Corollary}
\newtheorem{definition}{Definition}

\makeatletter
\newtheorem*{rep@theorem}{\rep@title}
\newcommand{\newreptheorem}[2]{%
\newenvironment{rep#1}[1]{%
 \def\rep@title{#2 \ref{##1}}%
 \begin{rep@theorem}}%
 {\end{rep@theorem}}}
\makeatother

\newreptheorem{theorem}{Theorem}
\newreptheorem{lemma}{Lemma}

\theoremstyle{definition}

\newcommand{\mO}{\mathcal{O}}

\newcommand{\Pclass}{\mathsf{P}}

\newcommand{\NP}{\mathsf{NP}}

\newcommand{\MA}{\mathsf{MA}}

\newcommand{\QMA}{\mathsf{QMA}}

\newcommand{\QCMA}{\mathsf{QCMA}}

\newcommand{\PP}{\mathsf{PP}}

\newcommand{\FNP}{\mathsf{FNP}}

\newcommand{\FBQP}{\mathsf{FBQP}}

\newcommand{\LH}{\textup{LH}}
\newcommand{\CLDM}{\textup{CLDM}}
\newcommand{\LEDMV}{\textup{LEDMV}}
\newcommand{\APXSIM}{\textup{APX-SIM}}

\newcommand{\SAT}{\textup{SAT}}

\newcommand{\poly}{\mathrm{poly}}

\newcommand{\kb}[2]{|#1\rangle\langle#2|}

\newcounter{protocolcounter}
\renewcommand{\theprotocolcounter}{\arabic{protocolcounter}}

\usepackage{mdframed}

\newenvironment{protocol}[1][]
{%
  \noindent \begin{minipage}{\textwidth}
    \begin{mdframed}[
      linewidth=0.8pt,
      roundcorner=5pt,
      backgroundcolor=white!,
      innertopmargin=10pt,
      innerbottommargin=10pt,
      innerleftmargin=10pt,
      innerrightmargin=10pt,
      skipabove=\topsep,
      skipbelow=\topsep,
    ]
      \refstepcounter{protocolcounter}
      \begin{center}
        \textbf{Protocol \theprotocolcounter:} #1
        \linebreak
      \end{center} 
}%
{%
    \end{mdframed}
  \end{minipage}
}

\newcounter{custalgocounter}
\renewcommand{\thecustalgocounter}{\arabic{custalgocounter}}

\usepackage{mdframed}

\newenvironment{custalgo}[1][]
{%
  \noindent \begin{minipage}{\textwidth}
    \begin{mdframed}[
      linewidth=0.8pt,
      roundcorner=5pt,
      backgroundcolor=white!,
      innertopmargin=10pt,
      innerbottommargin=10pt,
      innerleftmargin=10pt,
      innerrightmargin=10pt,
      skipabove=\topsep,
      skipbelow=\topsep,
    ]
      \refstepcounter{custalgocounter}
      \begin{center}
        \textbf{Algorithm \thecustalgocounter:} #1
        \linebreak
      \end{center} 
}%
{%
    \end{mdframed}
  \end{minipage}
}

\crefname{protocolcounter}{Protocol}{Protocols}
\Crefname{protocolcounter}{Protocol}{Protocols}

\crefname{custalgocounter}{Algorithm}{Algorithms}
\Crefname{custalgocounter}{Algorithm}{Algorithms}

\title{\huge Finding quantum partial assignments by search-to-decision reductions}

\author{Jordi Weggemans\thanks{Email: \href{mailto:jrw@cwi.nl}{jrw@cwi.nl}.}}
\affil{QuSoft \& CWI, Amsterdam, the Netherlands}

\begin{document}
\maketitle

\begin{abstract}
\noindent In computer science, many search problems are reducible to decision problems, which implies that finding a solution is as hard as deciding whether a solution exists. A quantum analogue of search-to-decision reductions would be to ask whether a quantum algorithm with access to a $\QMA$ oracle can construct $\QMA$ witnesses as quantum states. By a result from Irani, Natarajan, Nirkhe, Rao, and Yuen (CCC '22), it is known that this does not hold relative to a quantum oracle, unlike the cases of $\NP$, $\MA$, and $\QCMA$ where search-to-decision relativizes.

We prove that if one is not interested in the quantum witness as a quantum state but only in terms of its partial assignments, i.e.~the \textit{reduced density matrices}, then there exists a classical polynomial-time algorithm with access to a $\QMA$ oracle that outputs approximations of the density matrices of a near-optimal quantum witness, for any desired constant locality and inverse polynomial error. Our construction is based on a circuit-to-Hamiltonian mapping that approximately preserves near-optimal $\QMA$ witnesses and a new $\QMA$-complete problem, \textit{Low-energy Density Matrix Verification}, which is called by the $\QMA$ oracle to adaptively construct approximately consistent density matrices of a low-energy state.
\end{abstract}

\section{Introduction}
\label{sec:intro}
Decision (or promise) problems are arguably the central objects of study in computational complexity theory. While resolving a decision problem provides information about the \textit{existence} of a solution, it does not provide the solution itself. Fortunately, \emph{search problems}, where the task is to output an actual solution, are often reducible to their related decision problems. In this context, one generally considers \textit{Turing reductions}: here, one has access to an oracle capable of solving a class of decision problems, which is then used as a subroutine to solve the desired search problem.

As an example, consider a formula $\phi$ corresponding to a Boolean satisfiability ($\SAT$) problem on $n$ bits, and assume that we have access to an $\NP$ oracle. Under the assumption that $\phi$ is satisfiable, one can find a solution $x^*$ such that $\phi(x^*)=1$ in the following way: one queries the $\NP$ oracle adaptively to ask whether $\phi$ is satisfiable under the extra constraint that a certain subset of variables takes on specific values, i.e., under a fixed partial assignment. Every query to the oracle yields one bit of information about some $x^*$, and thus, after $n$ queries, the algorithm has found a solution.\footnote{It can return any satisfying assignment if the solution is not unique.} This strategy generally works for any problem in $\NP$ and can also be used to calculate the optimal value of an optimization problem up to exponential accuracy using binary search~\cite{krentel1986complexity}.

In~\cite{irani2022quantum}, Irani, Natarajan, Nirkhe, Rao, and Yuen studied whether a similar result holds in a quantum setting, where the goal is to output a \emph{quantum state} as a $\QMA$ witness, as opposed to a classical string. To extend the $\SAT$ example to the quantum case, one can consider its quantum generalization in terms of the local Hamiltonian problem ($\LH$). Here, the input is a Hermitian operator $H$ on $n$ qubits that can be efficiently written down as a sum of local terms, each acting non-trivially on only a subset of the qubits, and two parameters $a$ and $b$. The task is then to decide whether the ground state energy (its smallest eigenvalue) is $\leq a$ or $\geq b$. When $b-a = 1/\poly(n)$, the local Hamiltonian problem is $\QMA$-complete~\cite{Kitaev2002ClassicalAQ}. The question now is whether a quantum algorithm with access to a $\QMA$ oracle can prepare the ground state (the eigenstate corresponding to its smallest eigenvalue) of $H$ as a quantum state.

As pointed out in~\cite{irani2022quantum}, it seems unclear how to adapt the above strategy for $\NP$ to the local Hamiltonian problem (or any other $\QMA$-complete problem), because of the following two issues:
\begin{enumerate}[label=(\roman*)]
\item the description size complexity of a quantum state on $n$ qubits is generally exponential in $n$;
\item there does not appear to be a natural way of conditioning a quantum state on a partial assignment.
\end{enumerate}
It turns out that with a $\PP$-oracle, one can avoid this partial assignment strategy and generate $\QMA$ witnesses by making only a single quantum query~\cite{irani2022quantum}. Moreover,~\cite{irani2022quantum} shows that relative to a quantum oracle, $\QMA$ fails to have search-to-decision reductions, contrasting with some related classes where the witnesses are \emph{classical}. For instance, $\NP$, $\MA$, and $\QCMA$ all have search-to-decision reductions relative to all oracles. So what is possible with a $\QMA$-oracle?

\subsection{Results}
Going back to the local Hamiltonian problem, we observe that the full quantum state in fact contains more information than is needed; since the Hamiltonian is local, it suffices to have sufficiently good approximations of all $k$-local \emph{density matrices} of a low-energy state to compute the energy, provided we know that the density matrices are approximately \emph{consistent} with some global state. Constant-locality density matrices do not suffer from point (i) above, as there are only a polynomial number of them and each has a polynomially-sized description (for inverse polynomial accuracy). However, it is well-known that it is again $\QMA$-complete to check if all density matrices are consistent with a global quantum state~\cite{liu2007consistency,broadbent2022qma}.

We show that with access to a $\QMA$ oracle, a quantum analogue of the adaptive partial assignment strategy is possible for density matrices of low-energy states, which can be ensured to be approximately consistent. This demonstrates that point (ii) has a natural quantum manifestation for the class $\QMA$ when density matrices of low-energy states of local Hamiltonians are concerned. This is captured by the following (informal) theorem:
\begin{theorem}[Informal, from~\cref{thm:H_alg_dens} and~\cref{cor:H_alg_dens_deterministic}] For $k,q \in \mathbb{N}$ constant, we have that for any $k$-local Hamiltonian $H$, there exists a polynomial-time classical algorithm that makes queries to a $\QMA$ oracle and outputs a set of $q$-local density matrices that are at least arbitrarily (inverse-polynomially) close in trace distance to the density matrices of a state with energy arbitrarily (inverse-polynomially) close to the ground state energy.
\label{thm:1_informal}
\end{theorem}
Note that the algorithm works for any constant dimension of the density matrices, allowing one to store a classical fingerprint of a low-energy state that can be used to compute expectation values of observables up to this constant locality indefinitely. Density matrices seem to be the only type of classical witness we know of that serves as a correct classical fingerprint of the ground state (for all local observables) without imposing any additional structure on the ground state, such as being close to being samplable~\cite{gharibian2021dequantizing}, classically evaluatable, quantumly preparable~\cite{weggemans2023guidable}, succinct~\cite{jiang2023local}, etc., all of which would place the corresponding local Hamiltonian problem in $\QCMA$.\footnote{Or even smaller classes, depending on the class of states.} It is also straightforward to show that if the Hamiltonian has an inverse polynomially bounded spectral gap, the density matrices can be guaranteed to come from the actual ground state (\cref{cor:1}).

What about other problems in $\QMA$? With some more work, we show that the density matrices corresponding to a near-optimal witness for any problem in $\QMA$ can indeed also be found, as demonstrated in the following theorem.

\begin{theorem}[Informal, from~\cref{thm:main}] For any promise problem in $\QMA$, with input $x$ of size $n$ and verifier circuit $U_n$ using some polynomially-sized quantum proof $\xi$, and any $q$ constant, there exists a polynomial-time classical algorithm that makes queries to a $\QMA$ oracle which outputs:
\begin{itemize}
\item an arbitrarily (inverse-polynomially) good approximation of the maximum acceptance probability of $U_n$ on $(x,\xi)$ over all quantum proofs $\xi$.
\item A set of $q$-local density matrices whose elements are at least arbitrarily (inverse-polynomially) close in trace distance to the density matrices of a quantum proof $\xi$ which has an acceptance probability arbitrarily (inverse-polynomially) close to the maximum acceptance probability.
\end{itemize}
\label{thm:main_informal}
\end{theorem}

The key idea here is to use an approximately witness-preserving reduction from $\QMA$ verification circuits to the local Hamiltonian problem, as will be explained in~\cref{ss:proof_ideas}.

\paragraph{A new intuition.}
Whilst our results are not necessarily surprising, they have the merit of formalizing another intuition as to why quantum witnesses might not have search-to-decision reductions. That is, even though our approach to some extent circumvents the two issues (i) and (ii) from~\cref{sec:intro}, we have that now a single new issue that prevents search-to-decision for quantum witnesses\footnote{This was inspired by the introduction of~\cite{arad2024quasi}, where the ``bottom-up'' property is coined and discussed in some more detail.}:
\begin{itemize}
    \item Quantum states do not possess the \emph{``bottom-up'' property}; that is, given as an input (approximate) descriptions of all constant-locality density matrices that are (approximately) consistent with a global state, there does not appear to be any efficient procedure that allows you to construct the corresponding global state as a quantum state.
\end{itemize}
Classically, it is trivial to construct the global assignment if you are given a collection of consistent local assignments.\footnote{If the classical local assignments are approximate in the sense that each entry has its bit flipped with small probability, then for any probability upper bounded by a constant strictly smaller than $1/2$ you can make the locality of the marginals a large enough of constant so that comparing overlapping assignments with a majority vote leads to the correct global assignment with high probability (assuming you get all local assignments of a fixed locality, similar to our~\cref{thm:main_informal}).} 

Finally, we remark that the above issue is closely related to the $\QCMA$ versus $\QMA$ question. That is because, if such a procedure exists, it would directly imply that $\QCMA = \QMA$. In the {\sc yes}-case, the prover could provide descriptions of consistent density matrices, which the verifier uses to prepare the global state as a quantum witness. In the {\sc no}-case, it does not matter whether the procedure aborts on inconsistent density matrices or generates an arbitrary state, as both cases can be distinguished from the {\sc yes}-case. Since $\QCMA$ has search-to-decision reductions, this would also directly imply search-to-decision for $\QMA$.

\subsection{Proof ideas}
\label{ss:proof_ideas}
\paragraph{Finding low-energy marginals of local Hamiltonians.} We start by introducing a new $\QMA$-complete promise problem to be used by the $\QMA$ oracle, called the \textit{Low-energy Density Matrix Verification} ($\LEDMV$) problem. This problem can be viewed as a combination of the local Hamiltonian problem and the Consistency of Local Density Matrices ($\CLDM$) problem. One is given a $k$-local Hamiltonian $H$, a set of $q$-local density matrices $D = \{\rho_j\}$, and parameters $a$, $\delta$, $\alpha$, and $\beta$. The task is to decide whether there exists a state with energy $\leq a$ whose density matrices all have trace distance at most $\alpha$ from the corresponding density matrices in $D$, or if, for all states with energy less than $a + \delta$, there exists at least one density matrix in $D$ that has trace distance $\geq \beta$ from the corresponding density matrix of that state, promised that one of these conditions holds. This problem is trivially $\QMA$-hard since it reduces to the $\CLDM$ problem when $H = \mathbb{I}$, $a = 1$, and $\delta \geq 0$. Containment in $\QMA$ (\cref{lem:LEDMV_in_QMA}) can be shown by considering a protocol where the prover sends the classical descriptions of all reduced density matrices of some fixed locality of a certain low-energy state accompanied with a quantum proof. The verifier then checks whether these density matrices (i) are close to the ones in $D$, (ii) have low energy with respect to $H$, and (iii) are approximately consistent with a global state, using the quantum proof and the protocol for consistency of local density matrices~\cite{liu2007consistency}.

A probabilistic algorithm that constructs low-energy marginals for the local Hamiltonian can then be given as follows:
\begin{enumerate}
\item One finds a good estimate of the ground state energy using binary search (see also~\cite{ambainis2014physical,gharibian2019complexity}).
\item Next, one constructs the partial assignments of the density matrices by randomly guessing a partial assignment, using the previously obtained density matrices as an input, until a suitable one is found. For this, queries are made to a $\QMA$ oracle to solve instances of $\LEDMV$.
\end{enumerate}
This randomized algorithm can then be derandomized by replacing the random guessing with a brute-force search over an $\epsilon$-net of density matrices (see~\cref{subsec:derandomizaton}). Since we can make calls to the $\QMA$ oracle that are outside the promise set, one has to be careful about what happens when invalid queries are made. Crucially, by exploiting the fact that a {\sc yes} answer to an oracle call means that one can be sure it is \emph{not} a {\sc no} instance (even when an invalid query was made, see~\cite{gharibian2019complexity}), one can show that all iterations in step 2 increase the energy of the possible state of the density matrix as well as the error at every step, but this error can be made arbitrarily inverse polynomially small. Since the number of steps is only polynomial, the total error---both in terms of trace distance and energy---can be made inverse polynomially small as well.

\paragraph{Arbitrary problems in $\QMA$.} To obtain~\cref{thm:main_informal}, the key idea is to use an \emph{approximately witness-preserving reduction} from the $\QMA$-verification problem to a local Hamiltonian. To obtain precise bounds on the energy of the ground state and the maximal acceptance probability of the $\QMA$ verification circuit, we use the small-penalty clock construction of~\cite{Deshpande2020}. We prove that by fine-tuning the small-penalty parameter and using pre-idling on the circuit, any state with energy below a certain threshold must have overlap inverse polynomially close to one with a witness that has an acceptance probability inverse polynomially close to the maximum acceptance probability, tensored with a known state. The small penalty parameter in the clock construction gives the construction very precise control of the guarantees on the overlap and acceptance probability. This allows us to adopt the above algorithm for the corresponding local Hamiltonian problem to obtain approximations of the reduced density matrices at any constant locality of near-optimal witnesses for all problems in $\QMA$.

\subsection{Related work}
\paragraph{Queries to $\QMA$ oracles} In~\cite{ambainis2014physical}, Ambainis initiated the study of $\Pclass^{\QMA[\log]}$, where he showed that the problem $\APXSIM$ -- which formalizes the problem of computing expectation values of local observables on the ground state -- is complete for this class. This work was extended by Gharibian and Yirka~\cite{gharibian2019complexity}, who gave a similar $\Pclass^{\QMA[\log]}$-completeness result for estimating two-point correlation functions, as well as fixing a bug in the hardness proof of Ambainis' original work. In addition, Gharibian and Yirka showed that $\Pclass^{\QMA[\log]} \subseteq \PP$. In~\cite{gharibian2020oracle}, these types of ground state observable problems were studied for Hamiltonians under more physically motivated constraints.

A key difference between the setting in this work and the $\APXSIM$ problem is that, even though computing the density matrices can be viewed as computing the expectation values of many Pauli observables (viewing the density matrix in its Pauli decomposition), one needs to compute \emph{all} density matrices in a way such that they are consistent with a single global state all at once, which is not possible in their setting.

\paragraph{Search-to-decision in a quantum setting} Next to the work mentioned in the introduction by~\cite{irani2022quantum}, Gharibian and Kamminga study search-to-decision reductions for \emph{classical} problems using \emph{quantum} algorithms in~\cite{gharibian2024bqp}. Specifically, they examine this in the context of problems in $\NP$ where a quantum algorithm has access to an $\NP$ oracle. They show that $\FNP \subseteq \FBQP^{\NP[\log]}$, meaning that any witness to an $\NP$-relation can be found using a quantum algorithm that makes $\mO(\log n)$ $\NP$ queries.

As pointed out by Sevag Gharibian (private communication), a result similar to our~\cref{thm:1_informal} can be derived as a corollary of the proof that consistency is $\QMA$-hard under Turing reductions in~\cite{liu2007consistency}. This proof relies on techniques from convex optimization while treating consistency as a black-box constraint, and also identifies the density matrices corresponding to a low-energy state of a Hamiltonian. Roughly, the idea is that the local Hamiltonian problem can
be expressed as a convex program over consistent density matrices (which form a convex set), where the consistency constraint can be (approximately) evaluated using the $\QMA$ oracle. If there exists a low-energy state with energy below a certain input threshold, with high probability the convex optimization algorithm will find some description of the density matrices even if the oracle is ``imperfect''. The considered convex optimization algorithm in~\cite{liu2007consistency} outputs a candidate set of density matrices at each iteration, and cuts out a part of the search space depending on what was observed during the step. We argue that our construction is simpler and more directly aligned with the idea of adaptively constructing partial assignments.\footnote{It is also possible to find the descriptions of the density matrices bit-by-bit using a variation to our method, see the discussion on~\hyperref[final_remark]{page~\pageref{final_remark}}.}

\subsection{Open problems}
Of course, it remains open whether $\QMA$ has search-to-decision reductions that produce the actual quantum states corresponding to accepting witnesses. Since Kitaev's circuit-to-Hamiltonian mapping does not relativize, approaching this question in the local Hamiltonian setting is sensible, as it would directly bypass the quantum oracle separation found in~\cite{irani2022quantum}. In this direction, one could also explore whether imposing restrictions on the types of local Hamiltonians considered -- such as requiring them to be spectrally gapped, geometrically constrained, etc.~-- could simplify the problem, even if these Hamiltonians are not necessarily known to be $\QMA$-hard under these constraints.

Regarding our construction for finding the density matrices of $\QMA$ witnesses, an interesting open question is whether a circuit-to-Hamiltonian construction is necessary or if a direct approach using the trivial $\QMA$-complete problem of circuit verification could be sufficient. It is not clear if this would work, as it seems impossible to compute the acceptance probability of a verification circuit (which is a global observable) directly given only the density matrices of a quantum-proof as an input. This contrasts with the energy of local Hamiltonians, which can be decomposed into a sum of local observables.

\section{Preliminaries}

\paragraph{Notation}
For a Hamiltonian $H$, we say $\ket{\psi}$ is a ground state of $H$ if $\bra{\psi} H \ket{\psi} = \lambda_0$, where $\lambda_0 = \min_{\ket{\psi}} \bra{\psi} H \ket{\psi}$ is the ground state energy (i.e., the smallest eigenvalue) of $H$. The spectral gap of a Hamiltonian $H$ is defined as the difference between the two smallest eigenvalues (which can be zero if the ground space is degenerate). We denote $\mathbb{U}(d)$ as the unitary group of degree $d$, and $\mathbb{SU}(d)$ as the special unitary group (a normal subgroup of the unitary group where all matrices have determinant $1$). For a Hilbert space $\mathcal{H}$, let $\mathcal{D}(\mathcal{H})$ represent the set of all density matrices. We use $\norm{\cdot}_1$ to denote the trace norm. For a number $n \in \mathbb{N}$, write $[n] = \{1,2,\dots,n\}$ and let $[n]^k$ represent the set of all possible $k$-element subsets of $[n]$. For a subset $A \subseteq [n]$, we write $\overline{A}$ for the complementary subset, i.e., $\overline{A} = [n] \setminus A$.

\paragraph{Complexity theory}
We assume basic familiarity with complexity classes; for precise definitions, see the Complexity Zoo.\footnote{\url{https://complexityzoo.net/Complexity_Zoo}.} In this work, all quantum classes will be considered to be promise classes. For example, when we write $\QMA$, we implicitly mean $\mathsf{Promise}\QMA$. For a promise class $\mathcal{C}$, we denote $V^{\mathcal{C}}$ to indicate that a polynomial-time algorithm $V$ has access to an oracle for any problem $A = (A_\text{{\sc yes}},A_\text{{\sc no}}, A_\text{\sc inv})$ in $\mathcal{C}$. If $V$ makes invalid queries (i.e., $x \in A_\text{\sc inv}$), the oracle may respond arbitrarily with a {\sc yes} or {\sc no} answer~\cite{goldreich2006promise,gharibian2019complexity}.

\paragraph{Consistency of density matrices}
\noindent We will consider variants of the one-sided error consistency of local density matrices problem, first defined in~\cite{liu2007consistency}. 
\begin{definition}[Consistency of local density matrices ($\CLDM$)~\cite{liu2007consistency}] We are given a collection of local density matrices $\rho_1,\rho_2,\dots,\rho_m$, where each $\rho_i$ is a density matrix over qubits $C_i \subset [n]$, and $|C_i| \leq k$ for some constant $k$. Each matrix entry is specified by $\poly(n)$ bits of precision. In addition, we are given a real number $\gamma$ specified with $\poly(n)$ bits of precision. The problem is to distinguish between the following cases:
\begin{enumerate}
    \item There exists an $n$-qubit state $\sigma$ such that for all $i \in [m]$ we have $\norm{\tr_{\overline{C}_i}[\sigma] - \rho_i}_1 =0$.
    \item For all $n$-qubit states $\sigma$ there exists some $i \in [m]$ such that  $\norm{\tr_{\overline{C}_i}[\sigma] - \rho_i}_1 \geq \gamma$.
\end{enumerate} 
\label{def:CLDM}
\end{definition}
\begin{lemma}[Adapted from~\cite{liu2007consistency}]$\CLDM$ is in $\QMA$ for $\gamma = \Omega(1/\poly(n))$. 
\label{lem:CLDM_in_QMA}
\end{lemma}

Liu shows containment in $\QMA$ by giving a protocol in which the verifier performs a random Pauli measurement on a random subset $C_i$ qubits of the proof $\sigma$, which is then compared with what the expected outcome would be if the density matrix was equal to $\rho_i$. This only has a very small success probability, and using the relation $\QMA+=\QMA$ from~\cite{aharonov2003lattice} Liu shows that that success probability can be amplified using a form of parallel repetition without having to worry about entanglement across ``supposed copies'' of the proof. The two-sided error (so when there is an error parameter in case (i) in~\cref{def:CLDM}) is also known to be in $\QMA$ by a simple extension of the proof of~\cite{liu2007consistency}, see~\cite{buhrman2024quantum}.\footnote{This containment does come with some restrictions on how the completeness and soundness parameters can be related, which also depends on the locality $k$.} For hardness, \cite{liu2007consistency} also showed that $\CLDM$ is $\QMA$-hard under Turning reductions for $\gamma=  1/\poly(n)$. Later,~\cite{broadbent2022qma} proved that the two-sided error (so when there is also an error in case 1 in~\cref{def:CLDM}) is $\QMA$-hard with respect to Karp reductions for an inverse polynomial promise gap.  However, the one-sided error version of $\CLDM$ suffices for our purposes, and simplifies the analysis as we only need to specify a single parameter $\gamma$.

\section{Finding low-energy marginals of local Hamiltonians}
\label{sec:marginals_LH}
\subsection{A simple randomized algorithm}
Let us begin by defining a new promise problem called the \emph{Low-energy Density Matrix Verification} problem, which will serve as the $\QMA$-complete problem to be used by the oracle.

\begin{definition}[Low-energy Density Matrix Verification]$(\LEDMV(k,q,\delta,\alpha,\beta))$ Let $H = \sum_{i \in [m]} H_i$ be a $k$-local Hamiltonian on $n \in \mathbb{N}$ qubits of $m = \poly(n)$ terms $H_i$ which satisfy $0 \preceq H_i \preceq 1$, for all $i \in [m]$.  One is given efficient classical descriptions of parameters $a,\delta \geq 0$ as well as a description of a collection of $q$-local density matrices $D = \{\rho_j\}_{j\in [l]}$ with $l = \poly(n)$. For each $\rho_j$, let $\{ C_j \}$ be the set of sets of index labels of the qubits of $\rho_j$ and denote $\overline{C}_j = [n] \setminus C_j$ for the complementary subset. The task is to decide which of the following two cases hold, promised that either one is the case:
 \begin{enumerate}[label=(\roman*)]
        \item There exists an $n$-qubit state $\xi$ with $\tr[H \xi] \leq a$ such that 
$            \norm{ \tr_{\overline{C}_j} [\xi]-\rho_j }_1 \leq \alpha
       $ for all $j \in [l]$; 
         \item For all $n$-qubit states $\xi$ with $\tr[H \xi] \leq a + \delta$ we have that there exists an $j \in [l]$ such that
        $
             \norm{ \tr_{\overline{C}_j} [\xi]-\rho_j }_1 \geq \beta
        $.
    \end{enumerate}
\end{definition}

$\LEDMV$ is trivially $\QMA$-hard because one can choose the Hamiltonian to be the identity operator $\mathbb{I}$, set $a=m$, and let any $\delta \geq 0$, thereby reducing it to the $\QMA$-hard $\mathsf{CLDM}$ problem as defined in~\cref{def:CLDM}. To demonstrate containment, we will show that $\LEDMV$ is in $\QMA$ for a wide range of parameters. The $\QMA$ protocol is given in~\cref{prot1}.\\

\
\begin{protocol}[$\QMA$ protocol for $\LEDMV$.]
\noindent \textbf{Input:} $H$, $D$, $a$, $\delta$, $\alpha$, $\beta$.\\

\noindent \textbf{Set:} $\gamma := \min\{\frac{\delta}{m},\beta-\alpha\}$, $r := \max\{k,q\}$, $I := [n]^r$.\\

\noindent \textbf{Protocol:}
\begin{enumerate}
    \item The prover sends a classical description of the set $\Sigma := \{ \sigma_{i_1,\dots,i_r} \}_{(i_1,\dots,i_r) \in I}$ and a quantum proof $\xi_{\textup{proof}}$.
    \item Let $\{C^H_i\}$ be set of indices of qubits that terms $H_i$ acts on. The verifier performs the following four checks, and accepts if and only if all of them accept:
    \begin{itemize}
        \item \textbf{Check 1:} it checks if all $\sigma_{i_1,\dots,i_r}$ are valid density matrices.
        \item \textbf{Check 2:} it checks if $\sum_{i \in [m]} \max \tr[H_i \tr_{\overline{C}^H_i}[\sigma_{i_1,\dots,i_r}]] \leq a$, where the maximization is over all $\sigma_{i_1,\dots,i_r}$ that contain all indices in $C^H_i$.
        \item \textbf{Check 3:} it checks if $ \max \norm{\rho_j - \tr_{\overline{C}_j}[ \sigma_{i_1,\dots,i_r}]}_1 \leq \alpha$ for all $j \in [l]$, where the maximization is over all $\rho_{i_1,\dots,i_r}$ that contain all indices from $C_j$.
        \item \textbf{Check 4:} it uses the quantum proof $\xi_{\textup{proof}}$ to verify $\CLDM(\Sigma,\gamma)$, using the standard protocol as described in~\cite{liu2007consistency}.
    \end{itemize}
\end{enumerate}
  \label{prot1}
\end{protocol}

\

Let us first explain some notation and ideas behind~\cref{prot1}. Both the input Hamiltonian $H = \sum_{i \in [m]} H_i$ and the input set of density matrices $D = \{\rho_j\}_{j \in [l]}$ live in an overall Hilbert space comprising of $n$ qubits. There are two notions of locality, referring to the maximum number of qubits each term $H_i$ acts non-trivially on, denoted by $k$, and the maximum size of any set $C_j$, denoted by $q$, which contain all qubit indices on which a density matrix $\rho_j$ from the set $D$ is defined. To also be able to refer to the indices of the qubits a local term $H_i$ acts on, we define the set $\{C_i^H\}$ to play the same role for $H$ as ${C_j}$ does for $\rho_j$. We take $r = \max \{k,q\}$, such that when you take all $r$-local density matrices $\sigma_{i_1,\dots,i_r}$ of an $n$-qubit state $\xi$, which have indices from the set $I:=[n]^r$, you have all necessary information to evaluate the energy of $H$ and to compare the individual trace distances with density matrices from $D$.  However, there might be cases where different $\sigma_{i_1,\dots,i_r}$ both contain the same indices needed to evaluate some trace distance or energy, which might yield different values if the prover did not provide density matrices that are consistent. To work around this, we simply compute all of them, and take the maximum as our value to be used (see Checks 2 and 3). Note that this would not do anything if the prover is honest and provides a consistent collection of density matrices. Importantly, Check 4 is \emph{not} used to check if the density matrices from $D$ are consistent, but whether the density matrices provided by the prover are; if Check 4 succeeds, then Check 2 passing already gives you this information.

Let us now prove that~\cref{prot1} is sound. 
\begin{lemma} We have that $\LEDMV(k,q,\delta,\alpha,\beta)$ is in $\QMA$ for $k,q \in \mathbb{N}$ constant, $\beta-\alpha = \Omega(1/\poly(n))$ and $\delta = \Omega(1/\poly(n))$.  
\label{lem:LEDMV_in_QMA}
\end{lemma}
\begin{proof}[Proof]
We will prove the correctness of~\cref{prot1}. First, we argue that it can be performed in polynomial time, as the maximisation for each entry in the sum of Check $2$ requires a brute force computation over at most 
$\binom{n-k}{r-k}$
density matrices $\sigma_{i_1,\dots,i_q}$, as every subset of $k$ vertices in a complete hypergraph of degree $r$ is contained in that many edges. This binomial coefficient is polynomial in $n$ whenever $q$ and $k$ (and thus also $r$) are constant. A similar argument (with $q$ instead of $k$) can be made for Check 3. It is clear that all other steps must run in polynomial time for our choice of parameters.
\\
\linebreak
\noindent \textbf{Completeness:} This follows directly by providing all $r$-qubit reduced density matrices $\Sigma = \{\sigma_{i_1,\dots,i_r} | \sigma_{i_1,\dots,i_r} = \tr_{[n] \setminus \{i_1,\dots,i_r\}}[\xi ]], i_1,\dots,i_r \in I \} $, where $\xi$ is the state as in the promise of case (i).\footnote{The reader might correctly argue that such an exact description cannot be given using a polynomially-sized description, but an exponentially precise description can always be given which also suffices for our purposes as $\CLDM$ will also be in $\QMA$ when the $0$ in~\cref{def:CLDM} is replaced by something exponentially close to $0$.} Check 1 succeeds with certainty since all $\sigma_{i_1,\dots,i_r}$'s are density matrices; Check 2 and Check 3 also succeed with certainty because of the promise of being in a {\sc yes}-instance and the trace distance can only decrease under the partial trace and Check 4 succeeds w.h.p.~because of the arguments for Checks 1 and 2 and the fact that the prover provides exactly the density matrix descriptions of $\xi$, and $\CLDM(\{\sigma_i\},\gamma)$ is in $\QMA$ by~\cref{lem:CLDM_in_QMA}.\\

\noindent \textbf{Soundness:} We will use a proof by contradiction. Suppose Checks 1 up to and including 3 have already succeeded, which means that $\sum_{i \in [m]} \max \tr[H_i \tr_{\overline{C}^H_i}[\sigma_{i_1,\dots,i_r}]] \leq a$ and  $ \max \norm{\rho_j - \tr_{\overline{C}_j}[ \sigma_{i_1,\dots,i_r}]}_1 \leq \alpha$ for all $j \in [l]$. Now suppose that Check 4 accepts with probability $> 1/3$, then we have that there must exist a $\xi'$ such that for all $i_1,\dots,i_r \in I$ it holds that 
\begin{align*}
    \norm{\sigma_{i_1,\dots,i_r} - \tr_{[n] \setminus \{i_1,\dots,i_r\}} [\xi']}_1 < \gamma.
\end{align*}
However, this implies that 
\begin{align*}
    \tr[H \xi']&= \sum_{i \in [m]} \tr[H_i \tr_{\overline{C}_i^H} [\xi']]\\
    &=   \sum_{i \in [m]} \max \tr[H_i \left( \tr_{\overline{C}_i^H} [\xi']-  \tr_{\overline{C}^H_i}[\sigma_{i_1,\dots,i_r}]\right)] +\sum_{i \in [m]} \max  \tr[H_i \tr_{\overline{C}^H_i}[\sigma_{i_1,\dots,i_r}]]\\
    &\leq \sum_{i \in [m]} \max  \norm{ \left( \tr_{\overline{C}_i^H} [\xi']-  \tr_{\overline{C}^H_i}[\sigma_{i_1,\dots,i_r}]\right)}_1 + \sum_{i \in [m]} \max \tr[H_i \tr_{\overline{C}^H_i}[\sigma_{i_1,\dots,i_r}]]\\
    &< m \gamma + a\\
    &\leq a + \delta,
\end{align*}
for our choice of $\gamma$. Here we used (i) the linearity of the trace, (ii) that trace distance can only decrease under the partial trace and (iii) that the maximisation is performed over all $\sigma_{i_1,\dots,i_r}$ that contain all indices in $C^H_i$. At the same time,  using that Check 2 succeeded, for all $\rho_j \in D $ we must have
\begin{align*}
    \norm{\rho_j - \tr_{\overline{C}_j} [\xi']}_1 &\leq \max \norm{\rho_j -\tr_{\overline{C}_j}[ \sigma_{i_1,\dots,i_r}] }_1 + \max \norm{ \tr_{\overline{C}_j}[ \sigma_{i_1,\dots,i_r}]-  \tr_{\overline{C}_j} [\xi']}_1 \\
    &< \alpha + \gamma \\
    &\leq \beta.
\end{align*}
Hence, this implies that there must exist a state $\xi'$ with energy $< a+\delta$ such that all $\rho_j \in D $ are strictly less than $\beta$-consistent (in terms of trace distance) with $\xi'$, which is inconsistent with the promise in a {\sc no}-instance. Hence, Check 4 must reject with probability $\geq 2/3$, which means the overall procedure rejects with probability $\geq 2/3$.
\end{proof}

Our next step is to demonstrate that it is possible, for $q$ constant, to use a sampling procedure to efficiently find an approximation to any given $q$-qubit density matrix.

\begin{lemma} Let $\rho$ be any $q$-qubit density matrix for some constant $q \in \mathbb{N}$. Then there exists a polynomial-time randomized algorithm which outputs a density matrix $\hat{\rho}$ such that $\norm{\rho - \hat{\rho}}_1 \leq \epsilon$ with probability at least $\epsilon^{2(2^q-1)}$.
\label{lem:sampl_dm}
\end{lemma}
\begin{proof}

By~\cite{kus1988universality}, the probability density function of the squared fidelity $ y = |\langle \psi | \phi \rangle|^2 $ between two Haar random pure states $ \ket{\psi} $ and $ \ket{\phi} $ in a Hilbert space of dimension $ d $ is given by
\[
\mathbb{P}[|\langle \psi | \phi \rangle|^2 = y] = (d-1)(1-y)^{d-2}.
\]
Letting $ d = 2^q $ for a $ q $-qubit system, the cumulative distribution function for the squared fidelity being less than or equal to $ 1 - \epsilon^2 $ is
\[
\mathbb{P}[|\langle \psi | \phi \rangle|^2 \leq 1 - \epsilon^2 ] = \int_{0}^{1 - \epsilon^2} (d-1)(1-y)^{d-2} \, dy = 1 - \epsilon^{2(d-1)}.
\]
For pure states $ \ket{\psi} $ and $ \ket{\phi} $, the trace distance bound $ \|\ketbra{\psi} - \ketbra{\phi}\|_1 \leq \epsilon $ holds if and only if $ |\langle \psi | \phi \rangle|^2 \geq 1 - \epsilon^2 $. Therefore,
\[
\mathbb{P}\left[\|\ketbra{\psi} - \ketbra{\phi}\|_1 \leq \epsilon\right] = 1 - \mathbb{P}[|\langle \psi | \phi \rangle|^2 \leq 1 - \epsilon^2 ] = \epsilon^{2(d-1)}.
\]

For a $ q $-qubit density matrix $ \rho $, there exists a purification $ \ket{\xi} $ in a $ 2q $-qubit system. By Uhlmann’s Theorem, the fidelity between two density matrices equals the maximum fidelity between their purifications. Thus, sampling a Haar random pure state $ \ket{\phi} $ and considering the reduced density matrix $ \hat{\rho} $ on the first $ q $ qubits, we have that $
\mathbb{P}\left[\|\rho - \hat{\rho}\|_1 \leq \epsilon\right] \geq \epsilon^{2(2^q - 1)}$.
Sampling a Haar random unitary $ U \in \mathbb{U}(4^q) $ provides a description of $ \ket{\phi} $ and can be performed in polynomial time for constant $ q $ (e.g., \cite{ozols2009}). 
\end{proof}
We can now state the randomized $\QMA$ query algorithm to find all $q$-qubit marginals of a low-energy state of a $k$-local Hamiltonian in~\cref{alg1}.\\

\begin{custalgo}[$\QMA$-query algorithm to find $\epsilon$-approximations of the $q$-local density matrices of a low energy state of some $k$-local Hamiltonian $H$.]
\noindent \textbf{Input:} A classical description of all local terms of a Hamiltonian $H$, locality parameters $k,q$, an accuracy parameter $\epsilon$.\\ 

\noindent\textbf{Set:} $r := \max\{k,q\}$, $I := [n]^r$, $\alpha: = \epsilon/2$, $\beta:=\epsilon$, $T:=3 |I| \left( \frac{2}{\epsilon}\right)^{2(2^q-1)}$, $\delta := \frac{a}{|I|+1}$.\\

\noindent\textbf{Algorithm:} 
\begin{enumerate}
    \item Run a binary search on the local Hamiltonian problem corresponding to $H$ using the $\QMA$ oracle to find an estimate of $\hat{\lambda}_0$ such that $\lambda_0(H) \in [\hat{\lambda}_0-\delta, \hat{\lambda}_0 + \delta]$. Set $\{ a_l | a_l = \hat{\lambda}_0 + l \delta \}_{l \in [|I|]}$.
    \item Do the following at most $T$ times, starting with $l \leftarrow 1$: \\
    \phantom{        }Assume we are at step $l$ and have obtained $\{\rho_j\}_{j \in [l-1]}$.
    \begin{enumerate}
        \item \textbf{Partial assignment guess:} Guess a $q$-qubit density matrix $\rho_l$ in the following way: pick a Haar random unitary $U \in \mathbb{U}(4^q)$, create the corresponding Haar random pure state $\ket{\xi}$ by applying $U$ to the all-zeros state $\ket{0^{2q}}$ and trace out the last $q$ qubits to end up with a  $q$-qubit system described by a known density matrix $\rho$.
        \item \textbf{Partial assignment verification:} Make a single query to the $\QMA$ oracle with the instance $\LEDMV(k,q,\delta,\alpha,\beta)$ with $H$, $\{\rho_j\}_{j \in [l]}$ and $a_l$ as inputs. If the outcome is {\sc yes}, continue and set $l \leftarrow l+1$, $\rho_l = \rho$ and add $\rho_l$ to create the set $\{\rho_j\}_{j \in [l]}$. If the output is {\sc no}, return to step (a). 
    \end{enumerate}
    \item Output $\{\rho_j\}_{j \in |D|}$ (and optionally $\hat{\lambda}_0(H)$).
\end{enumerate}
  \label{alg1}
\end{custalgo}

\

\label{discussion_promise}
The key idea behind \cref{alg1} is that even density matrices \emph{within} the promise gap maintain sufficient precision for our desired approximation. This effectively creates a decision problem where the soundness parameter serves as an upper bound on precision. This concept stems from the nature of making oracle queries to promise problems: when you encounter a {\sc yes} instance, all you can be certain of is that it is \emph{not} a {\sc no} instance. However, it is crucial to demonstrate that enough samples are collected to ensure that, with high probability, only {\sc yes} instances could have been observed. Since density matrices are constructed through partial assignments, each step introduces a potential error. Therefore, one has to be careful to ensure that these errors remain small enough so that the state, which the density matrices approximately represent, does not significantly increase in energy.

\begin{theorem} Let $H = \sum_{i \in [m]} H_i$ be a $k$-local Hamiltonian on $n$ qubits of $m = \poly(n)$ terms $H_i$, $0 \preceq H_i \preceq 1$. Let $q \in \mathbb{N}$ some constant and $a \in [1/\poly(n),m]$ and $\beta = \Omega(1/\poly(n))$ be input parameters. Let $I = [n]^q$ and write $C_j$ for the $j$th element in $I$. Then there exists a randomized polynomial-time algorithm making queries to a $\QMA$ oracle which with probability $\geq 2/3$ outputs a set of $q$-local density matrices $\{\rho_j\}_{j \in I}$ for which there exists a $\xi$ which satisfies $\tr[H\xi] \leq \lambda_0 + a$,  such that for all $j \in [|I|]$ we have $
            \norm{\rho_j - \tr_{\overline{C}_j} [\xi]}_1 \leq \epsilon
$.
\label{thm:H_alg_dens}
\end{theorem}

\begin{proof} We will prove correctness and analyse the complexity of~\cref{alg1}.

\paragraph{Correctness:}
See~\cite{ambainis2014physical,gharibian2019complexity} for the correctness of Step 1. Since $\LEDMV$ is $\QMA$-complete, there is a polynomial-time Karp reduction from $\LH$ to $\LEDMV$, which can then be used to perform Step 1 as described in~\cite{ambainis2014physical,gharibian2019complexity}.

We have to show that we indeed have produced a set of density matrices that is approximately consistent with a low-energy state of $H$. To do this, we need to bound how much the energy of the obtained state grows as we collect more and more density matrices. Consider an arbitrary step $l$. If a query to the $\QMA$ oracle returns {\sc yes} for some sampled $\rho_l$, we can be certain that there exists a state $\xi_l$ with energy $\leq a_l +\delta$ such that $\norm{\rho_j - \tr_{\overline{C}_j} [\xi_l]}_1 \leq \epsilon$ for all $j \in [l]$. Let $\xi:= \xi_{|I|}$. Hence, for the last step ($l=|I|$) we must then have that $
    \tr[H\xi]  \leq a_{|I|} +\delta \leq \lambda_0 + a$,
for our choice of $\delta$, so $\xi$ is a low energy state with energy $\leq \lambda_0 + a$. We have that $
            \norm{\rho_j - \tr_{\overline{C}_j} [\xi]}_1 \leq \epsilon
$ for all $j \in [|I|]$ is trivially satisfied in the end, as at every intermediate value of $l$ it is guaranteed that $
            \norm{\rho_j - \tr_{\overline{C}_j} [\xi]}_1 \leq \epsilon
$ for all $j \in [l]$, as $\beta = \epsilon$. Therefore, all that is needed to ensure correctness is to prove that our choice for $T$ is large enough to succeed with high probability.

\paragraph{Complexity: } Step 1 makes $\mO(\log n)$ queries to the $\QMA$ oracle for any $\delta= \Omega(1/\poly(n))$. By~\cref{lem:sampl_dm}, we have that Step 2a of~\cref{alg1} samples a $q$-qubit reduced density matrix $\rho_j$ with trace distance $\leq \epsilon/2$ to $\tr_{\overline{C}_j}[\xi] $ with probability at least $\left( \frac{\epsilon}{2}\right)^{2(2^q-1)}$,
which means that 
\begin{align*}
    \mathbb{E}[\text{Number of samples until a single iteration of step 2 finishes}] \leq  \left( \frac{2}{\epsilon}\right)^{2(2^q-1)}.
\end{align*}
By linearity of the expectation value, we have that
\begin{align*}
    \mathbb{E}[\text{number of steps performed until~\cref{alg1} halts}] \leq |I| \left( \frac{2}{\epsilon}\right)^{2(2^q-1)}=:T'.
\end{align*}
By Markov's inequality, we can turn this into an algorithm which succeeds with probability $\geq 2/3$ by setting $T = 3T'$. Since $|I| = \mO(n^q)$ and $\epsilon = \Omega(1/\poly(n))$, the runtime is polynomial when $q \in \mO(1)$.
\end{proof}

It is easy to show that if  $H$ has a unique ground state and an inverse polynomially bounded spectral gap, then we can guarantee that~\cref{alg1} finds density matrices that come from a state that is close to the actual ground state.

\begin{corollary} Suppose $H$ has a unique ground state $\ket{\psi_0}$ with ground state energy $\lambda_0$ and spectral gap $\Delta = 1/\poly(n) $. Then under the same assumptions as~\cref{thm:H_alg_dens}, for any $\epsilon' = \Omega(1/\poly(n))$, $q \in \mO(1)$ there exists a randomized algorithm that makes queries to a $\QMA$ oracle which with probability $\geq 2/3$ outputs a set of $q$-local density matrices $\{\rho_j\}$ such that for all $j \in [|I|]$ we have that $
            \norm{\rho_j - \tr_{\overline{C}_j} [\ketbra{\psi_0}]}_1 \leq \epsilon'. $   
\label{cor:1}
\end{corollary}
\begin{proof}
We only need to show that parameter settings for $a$ and $\epsilon$ in~\cref{alg1} exist such that the corollary holds. We have that for any choice of $a = \Omega(1/\poly(n))$ and $\epsilon = \Omega(1/\poly(n))$, there exists a density matrix $\xi'$ such that for all $j \in [|I|]$ with energy $\lambda_0 + a$ we have
\[
    \norm{\sigma_j - \tr_{\overline{C}_j} [\xi']}_1 < \epsilon.
\]
Writing $H$ in its eigendecomposition, the spectral gap promise gives
\begin{align*}
    \tr(H \xi') &= \tr\left(\sum_{i} \lambda_i \ketbra{\psi_i} \xi'\right) \\
    &\geq \lambda_0 \tr(\Pi_0 \xi') + (\lambda_0 + \Delta) \tr\left(\left(\mathbb{I} - \Pi_0\right) \xi'\right) \\
    &= \lambda_0 + \Delta (1 - \tr(\Pi_0 \xi')),
\end{align*}
where $\Pi_0 = \ketbra{\psi_0}$.  Since the energy of $\xi'$ is at most $\lambda_0 + a$, we have
\[ 
    \lambda_0 + \Delta (1 - \tr(\Pi_0 \xi')) \leq \lambda_0 + a.
\]
Rearranging, we obtain 
\[ 
   \tr(\Pi_0 \xi') \geq 1 - \frac{a}{\Delta} .
\]
To bound the trace distance, we use our bound on $\tr(\Pi_0 \xi')$ and the relation between fidelity and trace distance to find
\begin{align*}
     \norm{\xi' - \ketbra{\psi_0}}_1 \leq 2 \sqrt{1-\tr(\Pi_0 \xi')} \leq 2 \sqrt{\frac{a}{\Delta}}.
\end{align*}
To relate this to the trace distance with any of the reduced states $\sigma_j$, we use the subadditivity of the trace norm and the fact that the trace distance cannot increase under the partial trace. For all $j \in [|I|]$, we then have 
\begin{align*}
    \norm{\tr_{\overline{C}_j} [\ketbra{\psi_0}] - \sigma_j}_1 &\leq \norm{\tr_{\overline{C}_j} [\ketbra{\psi_0}] - \tr_{\overline{C}_j} [\xi']}_1 + \norm{\tr_{\overline{C}_j} [\xi'] - \sigma_j}_1 \\
    &\leq \norm{\ketbra{\psi_0} - \xi'}_1 + \epsilon \\
    &\leq 2 \sqrt{\frac{a}{\Delta}}  + \epsilon.
\end{align*}
Finally, to satisfy the corollary, we set $\epsilon := \epsilon' / 2$ and choose $a = \epsilon'^2 \Delta / 16$. Then, the final trace distance can be bounded as
\[
    \norm{\tr_{\overline{C}_j} [\ketbra{\psi_0}] - \sigma_j}_1 \leq \frac{\epsilon'}{2} + \frac{\epsilon'}{2} = \epsilon',
\]
as desired.
\end{proof}
If the ground space is degenerate and has a gap $\Delta$ between the two smallest \emph{distinct} eigenvalues, we have that~\cref{cor:1} holds with respect to finding a state that is close to an arbitrary state in the ground space of $H$.

\subsection{Derandomization}
\label{subsec:derandomizaton}
The above construction can easily be derandomized by replacing the random sampling of unitaries with a brute-force search over a discretized set of local density matrices. We first introduce the notion of an $\epsilon$-covering set of density matrices.
\begin{definition}[$\epsilon$-covering set of density matrices] Let $\mathcal{H}$ be some $d$-dimensional Hilbert space. We say a discrete set of density matrices $D^d_\epsilon = \{\rho_i \} \subseteq \mathcal{D}(\mathcal{H})$ is $\epsilon$-covering for $\mathcal{D}(\mathcal{H})$ if for all $\sigma \in \mathcal{D}(\mathcal{H})$ there exists a $\rho_i \in D^d_\epsilon$ such that $\frac{1}{2}\norm{\rho_i - \sigma}_1 \leq \epsilon$.
\label{def:PEDMS}
\end{definition}
We will proceed by showing that, for any $\epsilon$ that is inverse polynomially small, one can construct such an $\epsilon$-covering set that is not too large.

\begin{lemma} Every $U \in \mathbb{SU}(2^n)$ can be implemented using $\mO(n^2 4^n)$ $\textup{CNOT}$ and $1$-qubit gates. 
\label{lem:2gatesuniversal}
\end{lemma}
For a proof, see Nielsen and Chang, Chapter 4~\cite{nielsen2010quantum}. By the Solovay-Kitaev theorem, one can approximate $U \in \mathbb{SU}(2)$ up to error $\epsilon$ in diamond norm using at most $\mO(\log^c(1/\epsilon) $ for some $c>1 $, using \textit{any} inverse-closed universal gate set. However, for our purposes we need the optimal scaling of $c=1$~\cite{harrow2002efficient}. However, many sets are known to exist that achieve this for $\mathbb{SU}(2)$, see for example~\cite{harrow2002efficient,ross2014optimal,forest2015exact,bocharov2015efficient,kliuchnikov2015practical,parzanchevski2018super}. Since all we care about is that the gates \textit{can} optimally efficiently approximate a unitary in  $\mathbb{SU}(2)$ and \textit{not} that one can also find the sequence efficiently, we simply use the gate set used in~\cite{harrow2002efficient} (which comes from~\cite{lubotzky1986hecke}).

\begin{lemma}[Adapted from~\cite{harrow2002efficient}] There exist a universal gate set  $\mathcal{G}$ with $|\mathcal{G}|=3$ such that for every $U \in \mathbb{SU}(2)$, there exists a circuit that uses only gates from $\mathcal{G}$ and approximates $U$ up to error $\epsilon$ in diamond norm using at most $\mO\left(\log(1/\epsilon) \right)$ gates.
\label{lem:SK_linear}
\end{lemma}

We now have all the necessary ingredients to give a method to construct $\epsilon$-covering sets of density matrices in polynomial time for any constant number of qubits.
\begin{lemma} For all $q \in \mathbb{N}$ constant, $0< \epsilon < 1$, there exists a polynomial-time algorithm that constructs a $\epsilon$-covering set of density matrices $D_\epsilon^{2^q}$ of size at most $\poly(1/\epsilon)$
in time $\poly(1/\epsilon)$.
\end{lemma}
\begin{proof}
Just as in~\cref{lem:sampl_dm}, we know that for each $q$-qubit density matrix $\rho$ there exists a $2q$-qubit purification $\ket{\xi}$ and that the fidelity between two density matrices is equal to the largest overlap between two purifications of those density matrices. Therefore, it suffices to create an $\epsilon$-net for ${4^{q}}$-dimensional pure states, which can be created by considering an approximation of $\mathbb{SU}(4^q)$. 

Let $\mathcal{G}'$ be the gate-set from~\cref{lem:SK_linear}, and $\mathcal{G}$ be the gate set which contains all gates from $\mathcal{G}'$ with the $\text{CNOT}$-gate added to it.  Note that the global phase is irrelevant when considering the density matrices, so it suffices to work only with $U \in \mathbb{SU}(4^{q})$. We construct the $\epsilon$-covering set of density matrices  $D^{2^q}_\epsilon = \{\rho_i\} $  such that  $\rho_i  = \tr_{B} \ket{\psi_i}\bra{\psi_i}$, where  $\ket{\psi_i} = U_i \ket{0 \dots 0}$ for an enumeration over all possible $U_i$ using a certain amount of gates from the set $\mathcal{G}$ such that any possible $U_i$ can be approximated up to error $\epsilon$. By~\cref{lem:2gatesuniversal}, we need at most $m := C_1 q^2 4^{2q}$ $\textup{CNOT}$s and $1$-qubit gates, where $C_1 > 0$ is some constant. Approximating the $1$-qubit gates with gates from $\mathcal{G}'$ and using that that errors in unitary approximation accumulate linearly,  we have that by~\cref{lem:SK_linear} the maximum needed circuit depth using the set $\mathcal{G}$ to approximate any $U \in \mathbb{SU}(4^{q})$ up to error $\epsilon$ can upper bounded as $C_2 m \log(m/\epsilon)$ for some constant $C_2 > 0$. Hence, using that $|\mathcal{G}|=4$, the total number of possible circuits can be upper bounded as
\begin{align*}
    \left(4 \binom{2q}{2} \right)^{C_2 m \log(m/\epsilon)} &\leq \left(16q^2\right)^{C_2 m \log(m/\epsilon)}\\
    &=\left(16q^2\right)^{C_2 m \log(m)} \left(1/\epsilon\right)^{C_2 m \log \left((16q^2\right)} = \poly(1/\epsilon)
\end{align*}
when $q$ is constant. Since we can efficiently enumerate over all these possible circuits (as there are only an inverse polynomial of them, we can efficiently generate $D^{2^q}_\epsilon$. This also implies that $|D^{2^q}_\epsilon| = \poly(1/\epsilon)$, as desired.
\end{proof}

We can now derandomize~\cref{alg1}, replacing the sampling in Step 2a by picking a density matrix from the set $D^{2^q}_\epsilon$, giving the following corollary. It is easy to modify the parameter $T$ in~\cref{alg1}  such that the criteria of the corollary below are met.

\begin{corollary} Let $H = \sum_{i \in [m]} H_i$ be a $k$-local Hamiltonian on $n$ qubits of $m = \poly(n)$ terms $H_i$, $0 \preceq H_i \preceq 1$. Let $q \in \mathbb{N}$ some constant and $a \in [1/\poly(n),m]$ and $\beta = \Omega(1/\poly(n))$ be input parameters. Let $I = [n]^q$ and write $C_j$ for the $j$th element in $I$. Then there exists a polynomial-time algorithm making queries to a $\QMA$ oracle which  outputs a set of $q$-local density matrices $\{\rho_j\}_{j \in I}$ for which there exists a $\xi$ which satisfies $\tr[H\xi] \leq \lambda_0 + a$,  such that for all $j \in [|I|]$ we have $
            \norm{\rho_j - \tr_{\overline{C}_j} [\xi]}_1 \leq \epsilon
$.
\label{cor:H_alg_dens_deterministic}
\end{corollary}

Note that this would also apply to~\cref{cor:1}.
\label{final_remark}
As a final remark, it seems also possible to modify~\cref{prot1} and~\cref{alg1} to find the entries of the density matrices on a bit-by-bit basis (this would also not require any randomness). To see this, note that~\cref{prot1} can easily be modified to instead work with partial descriptions of density matrices (where only some entries are specified up to a certain number of bits). However, this comes at the cost that $T$ in~\cref{alg1}, and thus $a_l$ in Step 2b, grows much larger as every time you move to a new partial assignment you incur an uncertainty error (see the discussion on~\hyperref[final_remark]{Page~\pageref{discussion_promise}}). However, since the number of steps $T$ would still be polynomial (when finding at most a polynomial number of bits per entry), one can simply choose a smaller inverse polynomial for $\gamma$ in~\cref{prot1} to ensure that the error does not grow to large.

\section{Finding marginals of near-optimal $\QMA$ witnesses}
\subsection{Approximately witness-preserving reductions in $\QMA$}
In this section, we demonstrate that density matrices of a near-optimal witness can be found for any problem in $\QMA$. The key idea involves using the Feynman-Kitaev circuit-to-Hamiltonian mapping~\cite{Kitaev2002ClassicalAQ} with a small penalty factor~\cite{Deshpande2020}, which transforms a quantum verification circuit $U_n$, consisting of $T$ gates from a universal set of at most $2$-local gates, which takes an input $x$ and a quantum witness $\ket{\psi} \in \left( \mathbb{C}^2\right)^{\otimes \poly(n)}$, into a $k$-local Hamiltonian of the form
\begin{align}
H^x_{\textup{FK}} = H_\text{in} + H_\text{clock} + H_\text{prop} +  \epsilon_\textup{penalty} H_\text{out},
\label{eq:H_FK}
\end{align}
where the locality $k$ depends on the specific construction used, and $\epsilon_{\textup{penalty}} > 0$. For our purposes, the exact form of these terms is not crucial, but for the $3$-local construction the precise descriptions can be found in~\cite{kempe20033local}.

The ground state of the first three terms, $H_0 := H_\text{in} + H_\text{clock} + H_\text{prop}$, is given by the so-called \textit{history states}, which have zero energy with respect to $H_0$ and are defined as
\begin{align}
    \ket{\eta(\psi)} = \frac{1}{\sqrt{T+1}} \sum_{t = 0}^T U_t \dots U_1 \ket{\psi} \ket{0} \ket{t},
    \label{eq:hist_state}
\end{align}
where $\ket{\psi} \in \left( \mathbb{C}^2\right)^{\otimes \poly(n)}$ is a quantum witness and $t$ represents the time step of the computation. It is easily verified that if $U_n$ accepts $(x, \ket{\psi})$ with probability $p$, then the corresponding history state has energy
\begin{align}
\bra{\eta(\psi)} H^x_{\textup{FK}} \ket{\eta(\psi)} = \epsilon_\textup{penalty} \frac{1-p}{T+1}.
\label{eq:energy_hs}
\end{align}
Moreover, by linearity, we have
\begin{align*}
     \alpha_1 \ket{\eta(\ket{\psi_1})} + \alpha_2  \ket{\eta(\ket{\psi_2})} =  \ket{\eta(\alpha_1 \ket{\psi_1} + \alpha_2 \ket{\psi_2})},
\end{align*}
so any linear combination of history states is in itself a history state. We will also need the following result on the spectral gap of $H_0$, proven in~\cite{aharonov2008adiabatic} (and probably other works). For completeness, we include a proof in~\cref{app:proof_gap} to avoid adopting the $\Omega(\cdot)$ notation from~\cite{aharonov2008adiabatic}.

\begin{restatable}{lemma}{spectralgap}
Suppose $H_\textup{clock}$ is chosen such that the history states are in the null space of $H_0$. Then $H_0$ has a spectral gap $\Delta$ satisfying $\Delta \geq  \frac{1}{(T+1)^2}$.
\label{lem:spec_gap_H0}
\end{restatable}
A key lemma from~\cite{Deshpande2020} demonstrates that, using the Schrieffer-Wolff transformation, one can determine precise energy intervals based on the acceptance probabilities of the verification circuit for the low-energy subspace of the Hamiltonian in~\cref{eq:H_FK}, provided that $\epsilon_\textup{penalty}$ is sufficiently small.\footnote{This lemma is not listed as a formal lemma in~\cite{Deshpande2020}, but can constructed from the text as found in Appendices A and B~\cite{Deshpande2020}.}

\begin{lemma}[Small-penalty clock construction, adopted from~\cite{Deshpande2020}] Let $U_n$ be a $\QMA$-verification circuit for inputs $x$, $|x|=n$, where $U_n$ consists of $T = \poly(n)$ gates from some universal gate-set using at most $2$-local gates. Denote $P(\psi)$ for the probability that $U_n$ accepts $(x,\ket{\psi})$, and let $H^x_{\textup{FK}}$ be the corresponding $3$-local Hamiltonian from the circuit-to-Hamiltonian mapping in~\cite{kempe20033local} with a $\epsilon_{\textup{penalty}}$-factor in front of $H_\text{out}$, as in~\cref{eq:H_FK}. Then for all $\epsilon_{\textup{penalty}} \leq \Delta/16$, we have that the low-energy subspace $\mathcal{S}_{\epsilon_\textup{penalty}}$ of $H$, i.e. $
    S_{\epsilon_\textup{penalty}} = \text{span} \{ \ket{\Phi} : \bra{\Phi} H \ket{\Phi} \leq {\epsilon_\textup{penalty}} \} $,
has that its eigenvalues $\lambda_i$ satisfy
\begin{align}
    \lambda_i \in \left[{\epsilon_\textup{penalty}} \frac{1-P(\psi_i)}{T+1} - c_0 \frac{\epsilon^2_{\textup{penalty}}}{\Delta}, {\epsilon_\textup{penalty}} \frac{1-P(\psi_i)}{T+1} + c_0 \frac{\epsilon^2_{\textup{penalty}}}{\Delta} \right],
    \label{eq:spclock}
\end{align}
for some universal constant $c_0 >0$.
\label{lem:spcc}
\end{lemma}
We will also use the following lemma, which shows that states with sufficiently low energy with respect to $H_\textup{FK}^{x}$ must be close to some history state.

\begin{lemma} Let $\ket{\Psi}$ be a state such that $\bra{\Psi} H_\text{FK}^{x} \ket{\Psi} \leq \delta$, where $H_{\textup{FK}}^{x}$ is given in~\cref{eq:H_FK} and let $\Delta$ be the spectral gap of $H_0$. Write $\Pi_\textup{hist}$ for the projector on the subspace spanned by all history states. Then $    \norm{\Pi_\textup{hist} \ket{\Psi}}^2_2 \geq 1- \frac{\delta}{\Delta}$.
\label{lem:overlap_hist}
\end{lemma}
\begin{proof}
Let $\{ \ket{\psi_i} \}$ be the eigenbasis of $H_0$, which consists of history states (spanning the null space of $H$) and non-history states (with energy at least $\Delta$). We can write $H_0 = H_0^0 + H_0^{\geq \Delta}$, where $H_0^0$ are all the terms in the spectral decomposition of $H$ with eigenvalues exactly zero and $H_0^{\geq \Delta}$ those with eigenvalues $\geq \Delta$. Since $H_\text{out}$ is PSD and $\epsilon_{\textup{penalty}} > 0$, we have
\begin{align*}
    \delta &\geq \bra{\Psi} H_\text{FK}^{x} \ket{\Psi} \\
    &\geq \bra{\Psi} H_0 \ket{\Psi}\\
    &= \bra{\Psi} H_0^0  \ket{\Psi} + \bra{\Psi} H_0^{\geq \Delta} \ket{\Psi} \\
    &= 0 + \bra{\Psi}  \sum_{i : \lambda_i \geq \Delta} \lambda_i \ket{\psi_i} \bra{\psi_i}   \ket{\Psi}\\
    & \geq \Delta \bra{\Psi}  \sum_{i : \lambda_i \geq \Delta}  \ket{\psi_i} \bra{\psi_i}   \ket{\Psi}\\
    &= \Delta \bra{\Psi}  (\mathbb{I} - \Pi_\textup{hist})   \ket{\Psi}\\
    &= \Delta \left(1-\norm{\Pi_\textup{hist} \ket{\Psi}}^2_2\right).
\end{align*}
Where we used that the history states span the ground state in $H_0$. The statement follows directly by rearranging the inequality.
\end{proof}

Now that we understand that states with low energies must have a significant overlap with the space spanned by history states, we aim to precisely characterize the maximum acceptance probability of the witness in this history state, given the state's energy relative to the ground state energy of $H^x_{\textup{FK}}$. This is addressed in the following lemma.

\begin{lemma} Let $p^*$ be the maximum acceptance probability of a $\QMA$ verification circuit. Let $H_{\textup{FK}}^{x}$ be the Hamiltonian as in~\cref{eq:H_FK} resulting from the small-penalty clock construction for some ${\epsilon_\textup{penalty}} < \Delta/16$, with ground state energy $ \lambda_0 ({\epsilon_\textup{penalty}})$. Suppose we are given a state $\ket{\Psi}$ with an energy at most 
\begin{align*}
     \lambda_0 ({\epsilon_\textup{penalty}}) + c_0 \frac{\epsilon^2_{\textup{penalty}}}{\Delta}.
\end{align*}
Then we have that $\ket{\Psi}$ has fidelity at least 
\begin{align}
    1-\left( \frac{\epsilon_\textup{penalty}}{\Delta}\frac{1-p^*}{T+1} + 2c_0 \left(\frac{\epsilon_\textup{penalty}}{\Delta}\right)^2 \right)
    \label{eq:thm_assumption}
\end{align}
 with a history state $\ket{\eta(\psi)}$ for some witness $\ket{\psi}$ which has an acceptance probability $\tilde{p}$ satisfying
\begin{align*}
     p^*-\tilde{p} &\leq  (T+1) 2 c_0 \frac{\epsilon_{\textup{penalty}}}{\Delta} + 2 (T+1) \sqrt{ \frac{\epsilon_\textup{penalty}}{\Delta}\frac{1-p^*}{T+1} + 2c_0 \left(\frac{\epsilon_\textup{penalty}}{\Delta}\right)^2 } \\
     &\quad  + \frac{\epsilon_\textup{penalty}}{\Delta}\frac{1-p^*}{T+1} + 2c_0 \left(\frac{\epsilon_\textup{penalty}}{\Delta}\right)^2.
\end{align*}
\label{lem:hist_state_overlap}
\end{lemma}
\begin{proof}
By~\cref{lem:spcc}, we have that the ground state energy of $H_\text{FK}^{x}$ satisfies
\begin{align*}
    \lambda_0 \in \left[{\epsilon_\textup{penalty}} \frac{1-p^*}{T+1} - c_0 \frac{\epsilon^2_{\textup{penalty}}}{\Delta}, {\epsilon_\textup{penalty}} \frac{1-p^*}{T+1} + c_0 \frac{\epsilon^2_{\textup{penalty}}}{\Delta} \right].
\end{align*}
Hence, we have that $\ket{\Psi}$ has an energy at most
\begin{align}
   \bra{\Psi} H_{\textup{FK}}^x \ket{\Psi} \leq {\epsilon_\textup{penalty}} \frac{1-p^*}{T+1} + 2c_0 \frac{\epsilon^2_{\textup{penalty}}}{\Delta}=:\delta.
   \label{eq:up_bound_Psi_energy}
\end{align}
Note the extra factor `$2$' incurred because of the theorem assumption (\cref{eq:thm_assumption}). We can write any state $\ket{\Psi}$ in the eigenbasis of $H_0$ as
\begin{align}
    \ket{\Psi} = \alpha \ket{\text{hist}} + \sqrt{1-\alpha^2} |\text{hist}^\perp \rangle,
    \label{eq:decomp_eigenbasis_H0}
\end{align}
for some real $\alpha \in [0,1]$, where $ \ket{\text{hist}}$ lives in the space spanned by the history states and $|\text{hist}^\perp \rangle$ in the space orthogonal to it. In its eigenbasis, $H_0$ is diagonal. Note that $\alpha^2 =  \norm{\Pi_\textup{hist} \ket{\Psi}}^2_2$. Hence, by~\cref{lem:overlap_hist} it must hold that 
\begin{align*}
    \alpha \geq \sqrt{1-\frac{\delta}{\Delta}} = \sqrt{1-\left( \frac{\epsilon_\textup{penalty}}{\Delta}\frac{1-p^*}{T+1} + 2c_0 \left(\frac{\epsilon_\textup{penalty}}{\Delta}\right)^2 \right)}.
\end{align*}
We expand the energy using the decomposition of $\ket{\Psi}$ in the eigenbasis of $H_0$ using~\cref{eq:decomp_eigenbasis_H0} as 
\begin{align*}
    \bra{\Psi}  H_{\textup{FK}}^x \ket{\Psi} &= \left(\alpha \bra{\text{hist}} + \sqrt{1-\alpha^2} \langle \text{hist}^\perp |\right) H_\text{FK}^{x} \left(\alpha \ket{\text{hist}} + \sqrt{1-\alpha^2} |\text{hist}^\perp \rangle\right)\\
    & = \alpha^2 \bra{\text{hist}} H_\text{FK}^{x} \ket{\text{hist}} +  \alpha \sqrt{1-\alpha^2} \bra{\text{hist}} H_\text{FK}^{x}  |\text{hist}^\perp \rangle \\
    &\qquad + \alpha \sqrt{1-\alpha^2}  \langle \text{hist}^\perp | H_\text{FK}^{x} \ket{\text{hist}} + (1-\alpha^2) \langle \text{hist}^\perp | H_\text{FK}^{x} |\text{hist}^\perp \rangle.
\end{align*}
We now want to find a lower bound on $\bra{\Psi}  H_{\textup{FK}}^x \ket{\Psi}$ in terms of $\bra{\text{hist}} H_\text{FK}^{x} \ket{\text{hist}}$ to compare with our upper bound in~\cref{eq:up_bound_Psi_energy}. To do this, we must find lower bounds on the other three terms in the expression. For the first one we have
\begin{align*}
     {\bra{\text{hist}} H_\text{FK}^{x}  |{\text{hist}^\perp}}\rangle =  {\bra{\text{hist}} H_0 + {\epsilon_\textup{penalty}} H_\text{out}  |\text{hist}^\perp \rangle}
     = {\epsilon_\textup{penalty}} {\bra{\text{hist}} H_\text{out}  |\text{hist}^\perp \rangle} 
     \geq -{\epsilon_\textup{penalty}},
\end{align*}
using that $\norm{H_\text{out}} \leq 1$ and that $ \bra{\text{hist}} H_0  |\text{hist}^\perp \rangle = 0$, which holds since $\ket{\text{hist}}$,$|\text{hist}^\perp \rangle$ live in separate eigenspaces of $H_0$.  Similarly, for the second term it must also hold that $\bra{\text{hist}^{\perp}} H  \ket{\text{hist}} \geq - {\epsilon_\textup{penalty}}$. And finally, for the third term we have {$
    \langle \text{hist}^\perp |  H_{\textup{FK}}^x |\text{hist}^\perp \rangle \geq \Delta \geq 0 $.}
Putting this all together, we have
\begin{align}
    \bra{\Psi} H_{\textup{FK}}^x \ket{\Psi}  &\geq \alpha^2 \bra{\text{hist}} _{\textup{FK}}^x \ket{\text{hist}}  - 2 \alpha \sqrt{1-\alpha^2} {\epsilon_\textup{penalty}},
    \label{eq:lb_HFK}
\end{align}
Suppose that $\ket{\text{hist}}$ encodes a witness with acceptance probability $\tilde{p}$ (recall that linear combinations of history states are also history states). We have that
\begin{align*}
    \bra{\text{hist}} H_{\textup{FK}}^x \ket{\text{hist}} = {\epsilon_\textup{penalty}} \frac{1-\tilde{p}}{T+1}. 
\end{align*}
Plugging this into~\cref{eq:lb_HFK} and combining the resulting expression with~\cref{eq:up_bound_Psi_energy} gives
\begin{align*}
    {\epsilon_\textup{penalty}} \frac{1-p^*}{T+1} + 2 c_0 \frac{\epsilon^2_{\textup{penalty}}}{\Delta} \geq \alpha^2 {\epsilon_\textup{penalty}} \frac{1-\tilde{p}}{T+1} - 2 \alpha \sqrt{1-\alpha^2} {\epsilon_\textup{penalty}}
\end{align*}
which after rearranging to get $p^*-\alpha^2\tilde{p}$ at the LHS of the inequality results in
\begin{align*}
     p^*-\alpha^2\tilde{p} \leq  (T+1) 2 c_0 \frac{\epsilon_{\textup{penalty}}}{\Delta}  + 2 \alpha (T+1)\sqrt{1-\alpha^2} + 1 - \alpha^2.
\end{align*}
which gives, using our bounds on $\alpha$ and the fact that $p^*-\alpha^2\tilde{p} \geq p^*-\tilde{p} $ as $p^* \geq \tilde{p} \geq 0 $ and $\alpha\in [0,1]$, as well as~\cref{lem:spec_gap_H0},
\begin{align*}
     p^*-\tilde{p} &\leq  (T+1) 2 c_0 \frac{\epsilon_{\textup{penalty}}}{\Delta} + 2 (T+1) \sqrt{ \frac{\epsilon_\textup{penalty}}{\Delta}\frac{1-p^*}{T+1} + 2c_0 \left(\frac{\epsilon_\textup{penalty}}{\Delta}\right)^2 } \nonumber\\
     &\quad  + \frac{\epsilon_\textup{penalty}}{\Delta}\frac{1-p^*}{T+1} + 2c_0 \left(\frac{\epsilon_\textup{penalty}}{\Delta}\right)^2,
\end{align*}
 which completes the proof.
\end{proof}

However, being close to a history state is insufficient for our purposes; we need to be close to an actual witness state $\ket{\psi}$ tensored with some other state we do not care about. We demonstrate that the standard technique of ``pre-idling'' the verification circuit~\cite{aharonov2008adiabatic, gharibian2021dequantizing, cade2023improved} ensures that all history states are close to a state of the form $\ket{\psi} \otimes \ket{\Phi}$, where $\ket{\Phi}$ is a known state.

\begin{lemma} Let $U_n = U_{n,T} \dots U_{n,1}$ be a $\QMA$ verification circuit that uses $T$ gates. Let $\tilde{U}_n =  \bar{U}_{n,T+M} \dots \bar{U}_{n,1}$ be the circuit which is as $U_n$ but with $M$ identities prepended to the circuit and $H_\text{FK}$ be the circuit-to-Hamiltonian mapping resulting from $\tilde{U}_n$. Then for any history state $\ket{\eta(\psi)}$ with witness $\ket{\psi}$, we have that there exists a state of the form $\ket{\psi} \otimes \ket{\Phi}$, where $\ket{\Phi}$ is known, which satisfies $
    | \bra{\eta(\psi) } \left(\ket{\psi} \otimes \ket{\Phi} \right)|^2 =  M/(M + T + 1)$.
\label{lem:pre_idling}
\end{lemma}
\begin{proof}
Consider the state $\ket{\psi} \otimes \ket{\Phi} $ with $
    \ket{\Phi} = \frac{1}{\sqrt{M}} \sum_{t = 0}^{M-1}  \ket{0 \dots 0}   \ket{t}$.
We have that the first $M$ gates $\bar{U}_t$ are all identities. A direct calculation shows
$
    | \bra{\eta(\psi) } \left(\ket{\psi} \otimes \ket{\Phi} \right)|^2 =  M/(M + T + 1)$.
\end{proof}
We are now prepared to integrate all the above and present our approximately witness-preserving reduction. This reduction enables us to approximate the highest-accepting witness by solving a local Hamiltonian problem.

\begin{theorem} 
Let $A$ be a promise problem in $\QMA$ and $x$, $|x|=n$, an input, with a $\QMA$ verification circuit $U_n$ using $T$ gates and has a witness register denoted by $W$. Suppose that $p^*$ is the maximum acceptance probability for $x$. Let $p_1(n),p_2(n)$ be any polynomially bounded functions that are $\geq 1$ for all $n \geq 1$ and set
\begin{align*}
         M : = \left(4 p_2(n)\right)^2 (T+1),  \quad {\epsilon_\textup{penalty}}:= \frac{1}{100 (c_0 + 1)(\tilde{T}+1)^4   \left(p_1(n) \cdot p_2 (n)\right)^2},
\end{align*}
where $\tilde{T} = M + T$. Then there exists a polynomial-time reduction from a $M$-pre-idled verification circuit $\tilde{U}_n$ with $\tilde{T} = M +T$ gates, to a local Hamiltonian $H$ such that for any state with $\ket{\Psi}$ that satisfies 
    \begin{align*}
         \bra{\Psi} H \ket{\Psi}\leq \lambda_0({\epsilon_\textup{penalty}}) + c_0 \epsilon^2_{\textup{penalty}} \tilde{T}^2 
    \end{align*}
 it holds that $
        \norm{\tr_{\overline{W}} \ketbra{\Psi} - \ketbra{\psi}}_1 \leq 1/2 p_2(n)$
    for some quantum witness $\ket{\psi}$, which satisfies the property that $U_n$ accepts $(x,\ket{\psi)}$ with probability at least $
        p^*-1/p_1(n)$.
\label{thm:approx_wit}
\end{theorem}
\begin{proof}
By~\cref{lem:pre_idling}, we can use pre-idling with $M$ gates, creating a new circuit $\tilde{U}_n$ with $\tilde{T} = M + T$ gates such that
\begin{align*}
      \norm{\ketbra{\eta(\psi)} -\ket{\psi}\ketbra{\Phi}{\psi}\bra{\Phi} }_1  &  = \sqrt{1-| \bra{\eta(\psi) } \left(\ket{\psi} \otimes \ket{\Phi} \right)|^2}\\
      &= \sqrt{1-\frac{M}{M + T + 1}}\\
      &\leq 1/4p_2(n)
\end{align*}
if $
  M \geq \left(4 p_2(n)-1\right)(T+1)$,
which is satisfied with our choice of $M$. The statement in the theorem then consequently holds since the trace distance can only decrease under taking the partial trace (taken over the non-witness registers). Hence, we can take $\tilde{T} = T + M = \poly(n)$ in the new circuit. By~\cref{lem:spec_gap_H0}, we have that the spectral gap $\Delta$ for our $H_0$ corresponding to the new circuit $\bar{U}$ satisfies  $\Delta \geq 1/(\tilde{T}+1)^2$. By~\cref{lem:hist_state_overlap} we have that for our choice of ${\epsilon_\textup{penalty}}$ that if we are given a state $\ket{\Psi}$ with energy at most  $
\lambda_0 ({\epsilon_\textup{penalty}}) + 2c_0 \epsilon^2_{\textup{penalty}} (\tilde{T}+1)^2$
then it has trace distance at least  
\begin{align*}
    \norm{\ketbra{\Psi} - \ketbra{\eta(\psi)}}_1 &=\sqrt{1-|\bra{\Psi}\ket{\eta(\psi)}|^2  }\\
    &\leq \sqrt{
     \frac{\epsilon_\textup{penalty}}{\Delta}\frac{1-p^*}{\tilde{T}+1} + 2c_0 \left(\frac{\epsilon_\textup{penalty}}{\Delta}\right)^2}\\
    &\leq   1/4 p_2 (n),
\end{align*}
 with a history state $\ket{\eta(\psi)}$ for some witness $\ket{\psi}$ with acceptance probability $\tilde{p}$ which satisfies
\begin{align*}
     p^*-\tilde{p} &\leq  (\tilde{T}+1) 2 c_0 \frac{\epsilon_{\textup{penalty}}}{\Delta} + 2 (\tilde{T}+1) \sqrt{ \frac{\epsilon_\textup{penalty}}{\Delta}\frac{1-p^*}{\tilde{T}+1} + 2c_0 \left(\frac{\epsilon_\textup{penalty}}{\Delta}\right)^2 } \\
     &\quad  + \frac{\epsilon_\textup{penalty}}{\Delta}\frac{1-p^*}{\tilde{T}+1} + 2c_0 \left(\frac{\epsilon_\textup{penalty}}{\Delta}\right)^2\\
     &\leq 1/p_1(n)
\end{align*}
as desired. Hence, we have by the triangle inequality
\begin{align*}
    \norm{\ketbra{\Psi} - \ket{\psi}\mathclose{\ketbra{\Phi}{\psi}}\mathclose{\bra{\Phi}}  }_1 &\leq \norm{\ketbra{\Psi} - \ketbra{\eta(\psi)}}_1 + \norm{\ketbra{\eta(\psi)} -\ket{\psi}\mathclose{\ketbra{\Phi}{\psi}}\mathclose{\bra{\Phi}} }_1 \\
    &\leq 1/2 p_2(n).
\end{align*}
The result directly follows since the trace distance can only be reduced by taking a partial trace.
\end{proof}
Note that in the above theorem we have left the dependence on $\epsilon_\textup{penalty}$ explicitly in the energy bound, since we do not know beforehand what $\lambda_0({\epsilon_\textup{penalty}})$ is going to be (even if we have set $\epsilon_\textup{penalty}$) as it depends on the maximum acceptance probability $p^*$. However, in the next section we will see that this is fine as we can estimate the ground state energy with $\QMA$ oracle access as shown in~\cref{sec:marginals_LH}.

Finally, we show that~\cref{thm:approx_wit} also holds in a mixed state setting, which will be important as~\cref{alg1} only returns density matrices that are promised to be approximately consistent with a global density matrix (and not a pure state). 
\begin{corollary} Under the same assumptions and parameter choices as~\cref{thm:approx_wit}, replacing $\ket{\Psi}$ with a mixed state $\xi$ such that 
    \begin{align*}
          \tr[H \xi ]\leq \lambda_0({\epsilon_\textup{penalty}}) + c_0 \epsilon^2_{\textup{penalty}} \tilde{T}^2,
    \end{align*}
 it holds that $
    \norm{\tr_{\overline{W}} [\xi] - \xi_{\textup{proof}}}_1 \leq 1/2 p_2(n)$
    for some quantum witness $\xi_{\textup{proof}}$, which satisfies the property that $U_n$ accepts $(x,\xi_{\textup{proof}})$ with probability at least $
        p^*-1/p_1(n)$.
        \label{cor:thm_red_mixed_state}
\end{corollary}
\begin{proof} For this proof, we will omit all super- and subscripts for the verification circuit $U_n$ and corresponding circuit-to-Hamiltonian mapping $H_{\textup{FK}}^x
$ (but they will be the same object as before). We assume that the verification circuit $U$ is already pre-idled as per~\cref{thm:approx_wit}. For this $U$, denote the $p(n)$-qubit proof register again as $W$. Suppose that the corresponding circuit-to-Hamiltonian mapping $H$ as per~\cref{thm:approx_wit} acts on $q(n)$-qubits. We consider another verification circuit $U_{\textup{ext}} = U \otimes \mathbb{I}$ with proof register $W_{\textup{ext}} = W \cup W'$, where $\mathbb{I}$ acts on an appended register $W'$ consisting of $q(n)$ qubits. It is easy to see that the corresponding circuit-to-Hamiltonian mapping of $U_{\textup{ext}}$ is of the form $H_{\textup{ext}} = H \otimes \mathbb{I}$, where $H$ is the circuit-to-Hamiltonian mapping from $U$ and $\mathbb{I}$ again acts on $q(n)$ qubits, which means that $H_{\textup{ext}}$ acts on $2q(n)$ qubits. Now suppose there exists a $q(n)$-qubit mixed state $\xi$ such that 
$\tr[H \xi ]\leq \lambda_0({\epsilon_\textup{penalty}}) + c_0 \epsilon^2_{\textup{penalty}} \tilde{T}^2 $. Then, there exists a $2q(n)$ purification $\ket{\Phi}$, $\tr_{W'}[\ketbra{\Phi} = \xi]$, such that 
\[
\tr[H_{\textup{ext}} \ketbra{\Phi} ]\leq \lambda_0({\epsilon_\textup{penalty}}) + c_0 \epsilon^2_{\textup{penalty}} \tilde{T}^2.
\]
 Moreover, as $H_{\textup{ext}}$ can be viewed as the circuit-to-Hamiltonian mapping from $U_{\textup{ext}}$,~\cref{thm:approx_wit} readily implies that there exists a $p(n)+q(n)$-qubit proof $\ket{\Psi}$ such that
 \[
\norm{\tr_{\overline{W_{\textup{ext}}}} [\ketbra{\Phi}] - \ketbra{\Psi}}_1 \leq 1/2 p_2(n), 
 \]
and $U_\textup{ext}$ accepts $(x,\ketbra{\Psi})$ with probability at least $p^*$. Taking the partial trace first over $\overline{W_{\textup{ext}}}$ and then over $\bar{W}'\setminus \overline{W_{\textup{ext}}}$, we end up with a state in the register $W$ again. Since $\ket{\Phi} $ is a purification of $\xi$, we have $\tr_{\overline{W}}[\xi] = \tr_{\overline{W}} [\ketbra{\Phi}]$. Since $\overline{W} \subset \overline{W_{\textup{ext}}}$, and the trace distance can only decrease under the partial trace, we have
\begin{align*}
    1/2 p_2(n) \geq  \norm{\tr_{\overline{W}} [\ketbra{\Phi}] -\tr_{\overline{W}}[ \ketbra{\Psi}]}_1 = \norm{\tr_{\overline{W}} [\xi] -\xi_{\textup{proof}}}_1.
\end{align*}
Here $\tr_{\overline{W}}[ \ketbra{\Psi}] =: \xi_{\textup{proof}}$ is a $p(n)$-qubit proof that $U$ accepts with probability at least $p^*$.
\end{proof}

\subsection{Finding marginals of  high-accepting $\QMA$ witnesses}
Finally, we can now combine the above ideas to show that for any problem in $\QMA$ the density matrices for a nearly optimal accepting witness can be obtained. We let $J$ be the set of all $q$-element subsets of the indices of the qubits on which $H_\textup{FK}^{x}$ is defined (which is not to be confused with the set $I$, which depends on $r$, i.e., the maximum of $q$ and $k$), and $J_W \subset J$ the set of all $q$-element index combinations of indices corresponding to the proof register. After we pre-idle the circuit $U_n$ and construct the corresponding $H_{\textup{FK}}^x$ for the some choice of $\epsilon_{\textup{penalty}}$, we simply run the~\cref{alg1} (randomized or derandomized) for $H_{\textup{FK}}^x$ to obtain all density matrices with indices from the set $J$ and finally keep only those with indices from $J_W$. The full algorithm is given in~\cref{alg2}.

\

\begin{custalgo}[$\QMA$ query algorithm to find approximations of the $q$-local density matrices of high-accepting witnesses.]
\noindent \textbf{Input:} $U_n$, $p_1$, $p_2$, $q$.\\

\noindent \textbf{Set:} $M$, $\tilde{T}$ and $\epsilon_\textup{penalty}$  as per~\cref{thm:approx_wit}, and set $a := c_0 (\tilde{T}+1)^2\epsilon^2_{\textup{penalty}} $, $\epsilon :=1/2p_2(n)$.\\

\noindent \textbf{Algorithm:} 
\begin{enumerate}
    \item Let $\tilde{U}_n$ be the $M$-pre-idled circuit of $U_n$.
    \item Construct $H_{\textup{FK}}^x$ for the choice of ${\epsilon_\textup{penalty}}$ according to~\cref{eq:H_FK}. Let $J$ be the set of all $q$-element subsets of qubits on which $H_{\textup{FK}}^x$ is defined and let $J_W$ be only those concerning the witness register $W$ of $\tilde{U}_n$.
    \item Run~\cref{alg1} (randomized or derandomized) for $H_{\textup{FK}}^x$ with $a$, $\epsilon$ to obtain $\{\rho_{i_1,\dots,i_q}\}_{i_1,\dots,i_q \in J}$ and $\hat{\lambda}_0(H_{\textup{FK}}^x)$.
    \item Output $\{\rho_{i_1,\dots,i_q}\}_{i_1,\dots,i_q \in J_W}$ and
$        \hat{p} :=   1- \frac{\hat{\lambda}_0(H_{\textup{FK}}^x) (\tilde{T}+1)}{{\epsilon_\textup{penalty}}}$.
\end{enumerate}
  \label{alg2}
\end{custalgo}

\

\begin{theorem} Let $A = (A_\text{yes},A_\text{no})$ be any problem in $\QMA$ having a uniform family of verifier circuits $\{U_n\}$ and let $x$, $|x|=n$ be the input. Then for any polynomially bounded functions $p_1(n)$, $p_2(n)$ that are $\geq 1$ for all $n \geq 1$, and any $q \in \mO(1)$ there exists a polynomial-time (randomized) algorithm that makes queries to a $\QMA$ oracle which outputs (with probability $\geq 2/3$)
\begin{itemize}
    \item A $\hat{p}$ which satisfies $\abs{p^*-\hat{p}} \leq 1/p_1 (n)$, where $p^*$ is the maximum probability that $U_n$ accepts $(x,\ket{\psi})$, where the maximum is over the witnesses $\ket{\psi} \in \left(\mathbb{C}^2\right)^{\otimes \poly(n)}$.
    \item A set of $q$-local density matrices $\{\rho_{i_q,\dots,i_q}\}$ whose elements are at least $1/p_2(n)$-close in trace distance to the density matrices of some $\xi_{\textup{proof}}$ which $\QMA$-verifier accepts with probability at least $\tilde{p} \geq p^*-1/p_1(n)$.
\end{itemize}
\label{thm:main}
\end{theorem}
\begin{proof}
We will prove that~\cref{alg2} satisfies the criteria of the theorem.
\paragraph{Correctness} Suppose $H_\textup{FK}^x$ acts on $p_3(n) = \poly(n)$ qubits. By~\cref{thm:H_alg_dens} we have that the density matrices $\{\rho_{i_1,\dots,i_q}\}_{i_1,\dots,i_q \in I}$ come from a state $\xi$ that has energy at most 
\begin{align*}
   \tr[H_\textup{FK}^x \xi] \leq \lambda_0 (H_\textup{FK}^x ) + a =  \lambda_0 (H_\textup{FK}^x ) + c_0 (\tilde{T}+1)^2\epsilon^2_{\textup{penalty}},
\end{align*}
satisfying the conditions of~\cref{thm:approx_wit} (and thus~\cref{cor:thm_red_mixed_state}). Therefore, we have 
 $   \abs{ \tilde{p}- p^*}  \leq 1/p_1(n) $ for some proof $\xi_{\textup{proof}}$. 
By~\cref{lem:spcc}, we have that the ground state energy estimate of $H_\text{FK}^{x}$ satisfies
\begin{align*}
    \hat{\lambda}_0(H_\textup{FK}^x) \in \left[{\epsilon_\textup{penalty}} \frac{1-p^*}{T+1}  \pm \left(c_0 \frac{\epsilon^2_{\textup{penalty}}}{\Delta} + \frac{a}{|I|+1}\right)\right]
\end{align*}
which implies 
\begin{align*}
    p^*\in \left[1 - \frac{\hat{\lambda}_0(H_\textup{FK}^x)  (\tilde{T} + 1)}{\epsilon_\textup{penalty}} \pm 2c_0 \epsilon_\textup{penalty} (\tilde{T} + 1)^2 \right]
\end{align*}
using our choice of $a$, the fact that $|I| \geq 1$ and  the bound on $\Delta$ from~\cref{lem:spec_gap_H0}. Since
\begin{align*}
    \hat{p} = 1 - \frac{\hat{\lambda}_0(H_\textup{FK}^x)  (\tilde{T} + 1)}{\epsilon_\textup{penalty}},
\end{align*}
we have that for our choice of $\epsilon_{\textup{penalty}}$,
\begin{align*}
    \abs{p^*-\hat{p}} &\leq 2c_0 \epsilon_\textup{penalty}(\tilde{T} + 1^3) \leq 1/p_1(n).
\end{align*}
Moreover, by~\cref{thm:H_alg_dens}, we know that~\cref{alg1} returns all $q$-local density matrices from qubits $J \supset J_W$, and all of them satisfy 
$
            \norm{\rho_{i_1,\dots,i_q} - \tr_{[p_3(n)]\setminus \{i_1,\dots,i_q\} } [\xi]}_1 \leq 1/2p_2(n),
$
which combined with~\cref{cor:thm_red_mixed_state} and the triangle inequality gives
\begin{align*}
    \norm{\rho_{i_1,\dots,i_q} - \tr_{[p_3(n)]\setminus \{i_1,\dots,i_q\}} [\xi]}_1 &\leq  \norm{\rho_{i_1,\dots,i_q}  - \tr_{[p_3(n)]\setminus \{i_1,\dots,i_q\}} [\xi]}_1  +  \\
            &\qquad \norm{\tr_{[p_3(n)]\setminus \{i_1,\dots,i_q\}} [\xi] - \tr_{[p_3(n)]\setminus \{i_1,\dots,i_q\}} [\xi_{\textup{proof}}] }_1 \\
            &\leq 1/p_2(n).
\end{align*} 

\paragraph{Complexity} The complexity is polynomial in $2^q$ and $1/\epsilon$. Since $\epsilon = 1/\poly(n)$ and $q \in \mO(1)$, the overall runtime is polynomial for both the randomized~(\cref{thm:H_alg_dens}) and derandomized version (\cref{cor:H_alg_dens_deterministic}).    
\end{proof}

\subsection*{Acknowledgements}
The author would like to thank Florian Speelman for helpful comments on an earlier draft, Sevag Gharibian for a pointer towards~\cite{liu2007consistency} and anonymous reviewers for useful feedback on the first version of this work. JW was supported by the Dutch Ministry of Economic Affairs and Climate Policy (EZK), as part of the Quantum Delta NL programme.

\bibliography{main.bib}

\newcommand{\etalchar}[1]{$^{#1}$}
\begin{thebibliography}{FGKM15}

\bibitem[ADK{\etalchar{+}}08]{aharonov2008adiabatic}
Dorit {Aharonov}, Wim~van Dam, Julia {Kempe}, Zeph {Landau}, Seth {Lloyd}, and Oded {Regev}.
\newblock Adiabatic quantum computation is equivalent to standard quantum computation.
\newblock {\em SIAM review}, 50(4):755--787, 2008.
\newblock \href{https://arxiv.org/abs/quant-ph/0405098}{\tt arXiv:quant-ph/0405098}.

\bibitem[Amb14]{ambainis2014physical}
Andris Ambainis.
\newblock On physical problems that are slightly more difficult than {QMA}.
\newblock In {\em 2014 IEEE 29th Conference on Computational Complexity (CCC)}, pages 32--43. IEEE, 2014.
\newblock \href{https://arxiv.org/abs/1312.4758}{\tt arXiv:1312.4758}.

\bibitem[AR03]{aharonov2003lattice}
Dorit Aharonov and Oded Regev.
\newblock A lattice problem in quantum {NP}.
\newblock In {\em 44th Annual IEEE Symposium on Foundations of Computer Science, 2003. Proceedings.}, pages 210--219. IEEE, 2003.

\bibitem[AS24]{arad2024quasi}
Itai Arad and Miklos Santha.
\newblock Quasi-quantum states and the quasi-quantum {PCP} theorem, 2024.
\newblock \href{https://arxiv.org/abs/2410.13549}{\tt arXiv:2410.13549}.

\bibitem[BG22]{broadbent2022qma}
Anne Broadbent and Alex~Bredariol Grilo.
\newblock {QMA}-hardness of consistency of local density matrices with applications to quantum zero-knowledge.
\newblock {\em SIAM Journal on Computing}, 51(4):1400--1450, 2022.
\newblock \href{https://arxiv.org/abs/1911.07782}{\tt arXiv:1911.07782}.

\bibitem[BHW24]{buhrman2024quantum}
Harry Buhrman, Jonas Helsen, and Jordi Weggemans.
\newblock Quantum pcps: on adaptivity, multiple provers and reductions to local hamiltonians.
\newblock {\em arXiv preprint arXiv:2403.04841}, 2024.

\bibitem[BRS15]{bocharov2015efficient}
Alex Bocharov, Martin Roetteler, and Krysta~M Svore.
\newblock Efficient synthesis of universal repeat-until-success quantum circuits.
\newblock {\em Physical review letters}, 114(8):080502, 2015.
\newblock \href{https://arxiv.org/abs/1404.5320}{\tt arXiv:1404.5320}.

\bibitem[CFG{\etalchar{+}}23]{cade2023improved}
Chris Cade, Marten Folkertsma, Sevag Gharibian, Ryu Hayakawa, Fran\c{c}ois Le~Gall, Tomoyuki Morimae, and Jordi Weggemans.
\newblock {Improved Hardness Results for the Guided Local Hamiltonian Problem}.
\newblock In {\em 50th International Colloquium on Automata, Languages, and Programming (ICALP 2023)}, volume 261, pages 32:1--32:19, 2023.
\newblock \href{https://arxiv.org/abs/2207.10250}{\tt arXiv:2207.10250}.

\bibitem[DGF22]{Deshpande2020}
Abhinav {Deshpande}, Alexey~V. {Gorshkov}, and Bill {Fefferman}.
\newblock {The importance of the spectral gap in estimating ground-state energies}.
\newblock {\em PRX Quantum}, 3(4):040327, December 2022.
\newblock \href{https://arxiv.org/abs/2207.10250}{\tt arXiv:2207.10250}.

\bibitem[FGKM15]{forest2015exact}
Simon Forest, David Gosset, Vadym Kliuchnikov, and David McKinnon.
\newblock Exact synthesis of single-qubit unitaries over clifford-cyclotomic gate sets.
\newblock {\em Journal of Mathematical Physics}, 56(8), 2015.
\newblock \href{https://arxiv.org/abs/1501.04944}{\tt arXiv:1501.04944}.

\bibitem[GK24]{gharibian2024bqp}
Sevag Gharibian and Jonas Kamminga.
\newblock {BQP}, meet {NP}: Search-to-decision reductions and approximate counting, 2024.
\newblock \href{https://arxiv.org/abs/2401.03943}{\tt arXiv:2401.03943}.

\bibitem[GLG22]{gharibian2021dequantizing}
Sevag Gharibian and Fran\c{c}ois Le~Gall.
\newblock Dequantizing the quantum singular value transformation: hardness and applications to quantum chemistry and the quantum {PCP} conjecture.
\newblock In {\em Proceedings of the 54th Annual ACM SIGACT Symposium on Theory of Computing}, STOC 2022, page 19–32, New York, NY, USA, 2022. Association for Computing Machinery.
\newblock \href{https://arxiv.org/abs/2111.09079}{\tt arXiv:2111.09079}.

\bibitem[Gol06]{goldreich2006promise}
Oded Goldreich.
\newblock On promise problems: A survey.
\newblock In {\em Theoretical Computer Science: Essays in Memory of Shimon Even}, pages 254--290. Springer, 2006.

\bibitem[GPY20]{gharibian2020oracle}
Sevag Gharibian, Stephen Piddock, and Justin Yirka.
\newblock Oracle complexity classes and local measurements on physical hamiltonians.
\newblock In {\em 37th International Symposium on Theoretical Aspects of Computer Science (STACS 2020)}. Schloss-Dagstuhl-Leibniz Zentrum f{\"u}r Informatik, 2020.
\newblock \href{https://arxiv.org/abs/1909.05981}{\tt arXiv:1909.05981}.

\bibitem[GY19]{gharibian2019complexity}
Sevag Gharibian and Justin Yirka.
\newblock The complexity of simulating local measurements on quantum systems.
\newblock {\em Quantum}, 3:189, 2019.
\newblock \href{https://arxiv.org/abs/1606.05626}{\tt arXiv:1606.05626}.

\bibitem[HRC02]{harrow2002efficient}
Aram~W Harrow, Benjamin Recht, and Isaac~L Chuang.
\newblock Efficient discrete approximations of quantum gates.
\newblock {\em Journal of Mathematical Physics}, 43(9):4445--4451, 2002.
\newblock \href{https://arxiv.org/abs/quant-ph/0111031}{\tt arXiv:quant-ph/0111031}.

\bibitem[INN{\etalchar{+}}22]{irani2022quantum}
Sandy Irani, Anand Natarajan, Chinmay Nirkhe, Sujit Rao, and Henry Yuen.
\newblock Quantum search-to-decision reductions and the state synthesis problem.
\newblock In {\em 37th Computational Complexity Conference (CCC 2022)}. Schloss Dagstuhl-Leibniz-Zentrum f{\"u}r Informatik, 2022.
\newblock \href{https://arxiv.org/abs/2111.02999}{\tt arXiv:2111.02999}.

\bibitem[Jia23]{jiang2023local}
Jiaqing Jiang.
\newblock Local {Hamiltonian} problem with succinct ground state is {MA}-complete, 2023.
\newblock \href{https://arxiv.org/abs/2309.10155}{\tt arXiv:2309.10155}.

\bibitem[KMH88]{kus1988universality}
M~Kus, J~Mostowski, and F~Haake.
\newblock Universality of eigenvector statistics of kicked tops of different symmetries.
\newblock {\em Journal of Physics A: Mathematical and General}, 21(22):L1073, 1988.

\bibitem[KMM15]{kliuchnikov2015practical}
Vadym Kliuchnikov, Dmitri Maslov, and Michele Mosca.
\newblock Practical approximation of single-qubit unitaries by single-qubit quantum {Clifford} and {T} circuits.
\newblock {\em IEEE Transactions on Computers}, 65(1):161--172, 2015.
\newblock \href{https://arxiv.org/abs/1212.6964}{\tt arXiv:1212.6964}.

\bibitem[KR03]{kempe20033local}
Julia Kempe and Oded Regev.
\newblock {3-local Hamiltonian is QMA-complete}.
\newblock {\em Quantum Inf. Comput.}, 3:258--264, 2003.
\newblock \href{https://arxiv.org/abs/quant-ph/0302079}{\tt arXiv:quant-ph/0302079}.

\bibitem[Kre86]{krentel1986complexity}
Mark~W Krentel.
\newblock The complexity of optimization problems.
\newblock In {\em Proceedings of the eighteenth annual ACM symposium on Theory of computing}, pages 69--76, 1986.

\bibitem[KSV02]{Kitaev2002ClassicalAQ}
Alexei~Y. Kitaev, Alexander Shen, and Mikhail~N. Vyalyi.
\newblock {\em Classical and quantum computation}.
\newblock American Mathematical Society, 2002.

\bibitem[Liu06]{liu2007consistency}
Yi-Kai Liu.
\newblock Consistency of local density matrices is {QMA}-complete.
\newblock In {\em Approximation, Randomization, and Combinatorial Optimization. Algorithms and Techniques: 9th International Workshop on Approximation Algorithms for Combinatorial Optimization Problems, APPROX 2006 and 10th International Workshop on Randomization and Computation, RANDOM 2006, Barcelona, Spain, August 28-30 2006. Proceedings}, pages 438--449. Springer, 2006.
\newblock \href{https://arxiv.org/abs/quant-ph/0604166}{\tt arXiv:quant-ph/0604166}.

\bibitem[LPS86]{lubotzky1986hecke}
Alexander Lubotzky, Ralph Phillips, and Peter Sarnak.
\newblock Hecke operators and distributing points on the sphere i.
\newblock {\em Communications on Pure and Applied Mathematics}, 39(S1):S149--S186, 1986.

\bibitem[NC10]{nielsen2010quantum}
Michael~A Nielsen and Isaac~L Chuang.
\newblock {\em Quantum computation and quantum information}.
\newblock Cambridge university press, 2010.

\bibitem[Ozo09]{ozols2009}
Maris Ozols.
\newblock How to generate a random unitary matrix, 2009.

\bibitem[PS18]{parzanchevski2018super}
Ori Parzanchevski and Peter Sarnak.
\newblock Super-golden-gates for {PU(2)}.
\newblock {\em Advances in Mathematics}, 327:869--901, 2018.
\newblock \href{https://arxiv.org/abs/1704.02106}{\tt arXiv:1704.02106}.

\bibitem[RS16]{ross2014optimal}
Neil~J. Ross and Peter Selinger.
\newblock Optimal ancilla-free {Clifford}+{$T$} approximation of $z$-rotations.
\newblock {\em Quantum Info. Comput.}, 16(11–12):901–953, sep 2016.
\newblock \href{https://arxiv.org/abs/1403.2975}{\tt arXiv:1403.2975}.

\bibitem[Spi09]{Spielman2009}
Dan Spielman.
\newblock Spectral graph theory lecture notes, Fall 2009.

\bibitem[WFC23]{weggemans2023guidable}
Jordi Weggemans, Marten Folkertsma, and Chris Cade.
\newblock Guidable local {Hamiltonian} problems with implications to heuristic ans\"{a}tze state preparation and the quantum {PCP} conjecture, 2023.
\newblock \href{https://arxiv.org/abs/2302.11578}{\tt arXiv:2302.11578}.

\end{thebibliography}
\bibliographystyle{alpha}

\appendix

\section{Proof of~\cref{lem:spec_gap_H0}}
\label{app:proof_gap}
\spectralgap*
\begin{proof}
We follow~\cite{Kitaev2002ClassicalAQ} to inspect the spectrum of $H_\textup{prop}$. Applying a basis transformation of $W = \sum_{t=0}^{T} U_t \dots U_1 \otimes \ket{j} \bra{j}$ to $H_\text{prop}$ gives us 
    \begin{align*}
        W^{\dagger} H_\text{prop} W  = \sum_{t = 1}^T I \otimes E_t = I \otimes E
    \end{align*}
    where $E_t = \frac{1}{2} \left( - \kb{t}{t-1} -\kb{t-1}{t} + \kb{t}{t} + \kb{t-1}{t-1} \right)$ and thus
    \begin{align*}
        E = \begin{pmatrix}
            \frac{1}{2} & -  \frac{1}{2} & &  &  & 0 \\
            - \frac{1}{2} & 1  & - \frac{1}{2} &  &  & \\
           & - \frac{1}{2} & 1  & - \frac{1}{2}  &  &  \\
          &   & - \frac{1}{2} &  \ddots  & \ddots  & \\
          &   & &  \ddots &  1  & - \frac{1}{2}\\
           0 &   & &  &  -\frac{1}{2} & \frac{1}{2}
        \end{pmatrix}.
    \end{align*}
$E$ is the Laplacian of random walk on a line with $T+1$ nodes, which has eigenvalues
\begin{align*}
    \lambda_k = 1 - \cos\left( \frac{\pi k}{T+1}\right),
\end{align*}
with $0 \leq k \leq L$~\cite{Spielman2009}. Hence, its smallest non-zero eigenvalue is lower bounded by
\begin{align*}
    \lambda_1 (E) \geq 1 - \cos\left( \frac{\pi }{T+1}\right) \geq \frac{1}{3} \left( \frac{\pi }{T+1}\right) ^2 \geq \frac{1}{(T+1)^2}.
\end{align*}
Write $\mathcal{N}(H_0)$ and $\mathcal{N}^{\perp}(H_0)$ for the null space of $H_0$ and the space orthogonal to it. Since the null space of $H_0$ is spanned by history states~\cite{Kitaev2002ClassicalAQ}, we have that for any state $\ket{\phi} \in \mathcal{N}^{\perp}(H_0)$ it must hold that
\begin{align*}
    \bra{\phi} H_0 \ket{\phi} =  \bra{\phi} H_\text{in} \ket{\phi} + \bra{\phi} H_\text{clock}  \ket{\phi} + \bra{\phi}  H_\text{prop} \ket{\phi} \geq \lambda_1(E) \geq \frac{1}{(T+1)^2},
\end{align*}
using that $H_\textup{clock}$ and $H_\textup{in}$ are PSD and $H_\textup{prop}$ has the same smallest non-zero eigenvalue as $E$ (as the spectrum is preserved under basis transformations). Hence, the spectral gap $\Delta$ of $H_0$ satisfies  $
    \Delta \geq \frac{1}{(T+1)^2}$,
completing the proof.
\end{proof}

\end{document}